\newif\ifsubmission
\newif\iffull
\newif\ifready
\readytrue
\newif\ifnotready
\notreadyfalse
\fulltrue
\submissionfalse

\iffull
\documentclass[11pt]{article}
\fi
\ifsubmission
\documentclass[a4paper,UKenglish,autoref, thm-restate]{lipics-v2021}
\fi

\usepackage[ruled,vlined, linesnumbered]{algorithm2e}
\usepackage{amsmath}
\usepackage{amssymb}
\usepackage{amsthm}
\usepackage{booktabs}
\usepackage[utf8]{inputenc}
\usepackage{csquotes}
\iffull
\usepackage[margin=2cm]{geometry}
\fi
\usepackage{hyperref}
  \hypersetup{colorlinks={true},linkcolor={blue},citecolor=blue}
  \usepackage[capitalize, nameinlink]{cleveref}
\usepackage{mathtools}
\usepackage{siunitx}
\usepackage{tikz}
\usepackage{xcolor}
\usepackage{subcaption}
\usepackage{layouts}

\iffull
\usepackage{placeins}
\fi

\ifsubmission
\EventEditors{Petra Berenbrink and Benjamin Monmege}
\EventNoEds{2}
\EventLongTitle{39th International Symposium on Theoretical Aspects of Computer Science (STACS 2022)}
\EventShortTitle{STACS 2022}
\EventAcronym{STACS}
\EventYear{2022}
\EventDate{March 15--18, 2022}
\EventLocation{Marseille, France}
\EventLogo{}
\SeriesVolume{219}
\ArticleNo{12}
\fi

\newcounter{section-preserve}
\newcounter{theorem-preserve}
\newcommand{\blank}[1]{}
\newtoks\magicAppendix
\magicAppendix={}
\newtoks\magictoks
\newif\iflater
\laterfalse
\long\def\later#1{\magictoks={#1}
  \edef\magictodo{\noexpand\magicAppendix={\the\magicAppendix \par
    \the\magictoks
  }}
  \magictodo}
\long\def\both#1{\magictoks={#1}
  \edef\magictodo{\noexpand\magicAppendix={\the\magicAppendix \par
    \noexpand\setcounter{theorem-preserve}{\noexpand\arabic{theorem}}
    \noexpand\setcounter{theorem}{\arabic{theorem}}
    \noexpand\setcounter{section-preserve}{\noexpand\arabic{section}}
    \noexpand\setcounter{section}{\arabic{section}}
    \noexpand\let\noexpand\oldsection=\noexpand\thesection
    \noexpand\def\noexpand\thesection{\thesection}
    \noexpand\let\noexpand\oldlabel=\noexpand\label
    \noexpand\let\noexpand\label=\noexpand\blank
    \the\magictoks
    \noexpand\setcounter{theorem}{\noexpand\arabic{theorem-preserve}}
    \noexpand\setcounter{section}{\noexpand\arabic{section-preserve}}
    \noexpand\let\noexpand\thesection=\noexpand\oldsection
    \noexpand\let\noexpand\label=\noexpand\oldlabel
  }}
  \magictodo
  \the\magictoks}
\def\magicappendix{\latertrue \the\magicAppendix}

\ifsubmission
\else
\newtheorem{theorem}{Theorem}
\newtheorem{definition}[theorem]{Definition}
\newtheorem{lemma}{Lemma}

\fi
\newtheorem{challenge}{Challenge}

\newcommand{\myparagraph}[1]{\noindent {\bf #1.}}

\newcommand{\eps}{\varepsilon}
\newcommand{\pred}{\mathsf{Pred}}

\newcommand{\successor}{\mathsf{Succ}}
\newcommand{\poly}{\text{poly}}

\newcommand{\hash}{\mathcal{H}}

\newcommand{\tO}{\tilde{O}}

\newcommand{\dist}{\mathcal{D}}
\newcommand{\est}{\mathsf{CountDistinctEstimator}}
\newcommand{\subgraphset}{\mathcal{H}}

\newcommand{\batch}{B}
\newcommand{\bucket}{K}
\newcommand{\defn}[1]{\textbf{\emph{#1}}}

\newcommand{\Repeating}{Stale\xspace}
\newcommand{\Bad}{\Repeating}

\newcommand{\sgraph}[1]{G_{#1}}
\newcommand{\sedge}[1]{\EdgeSet{#1}}
\newcommand{\sgraphh}[2]{G_{#2}(#1)}
\newcommand{\sedgeh}[2]{\EdgeSetH{#1}{#2}}
\newcommand{\sv}[1]{\AncestorSet{#1}}
\newcommand{\svh}[2]{\AncestorSetH{#1}{#2}}
\newcommand{\smallparagraph}[1]{\noindent\emph{#1}.}

\newcommand{\Mergeable}{Mergeable\xspace}
\newcommand{\estv}[1]{\hat{a}(#1)}
\newcommand{\ev}[1]{\hat{a}(#1)}
\newcommand{\este}[1]{\hat{e}(#1)}
\newcommand{\ee}[1]{\hat{e}(#1)}
\newcommand{\estvh}[2]{\hat{a}_{#2}(#1)}
\newcommand{\esteh}[2]{\hat{e}_{#2}(#1)}

\newcommand{\ignore}[1]{}
\ifnotready
\newcommand{\qq}[1]{{\color{purple} Quanquan: #1}}
\newcommand{\zoya}[1]{{\color{olive} Zoya: #1}}
\newcommand{\mnote}[1]{{\color{red} MP: #1}}
\newcommand{\josh}[1]{{\color{cyan} Josh: #1}}
\newcommand{\enote}[1]{{\color{violet} Erik: #1}}
\fi

\ifready
\newcommand{\qq}[1]{}
\newcommand{\zoya}[1]{}
\newcommand{\mnote}[1]{}
\newcommand{\josh}[1]{}
\newcommand{\enote}[1]{}
\fi

\newcommand{\good}{fresh\xspace}
\newcommand{\bad}{stale\xspace}
\newcommand{\AncestorSet}[1]{\mathcal{A}(#1)}
\newcommand{\AncestorSetH}[2]{\mathcal{A}_{#2}(#1)}
\newcommand{\EdgeSet}[1]{\mathcal{E}(#1)}
\newcommand{\EdgeSetH}[2]{\mathcal{E}_{#2}(#1)}
\newcommand{\as}[1]{\AncestorSet{{#1}}}

\newcommand{\es}[1]{\EdgeSet{{#1}}}

\newcommand{\CommunicationDelay}{\rho}
\newcommand{\NumberVertices}{n}
\newcommand{\NumberEdges}{m}
\newcommand{\NumberMachines}{M}

\newcommand{\mergeable}{mergeable\xspace}
\newcommand{\estimationtime}{\ln^2 n}

\newcommand{\cS}{\mathcal{S}}

\newcommand{\opt}{\mathsf{OPT}}
\renewcommand{\emptyset}{\varnothing}

\title{Scheduling with Communication Delay in Near-Linear Time}
\ifsubmission
\author{Quanquan C. Liu}{MIT CSAIL \and \url{https://quanquancliu.com/}}{quanquan@mit.edu}{}{}
\author{Manish Purohit}{Google Research \and \url{https://sites.google.com/view/manishpurohit/home}}{purohitmanish89@gmail.com}{}{}
\author{Zoya Svitkina}{Google Research \and \url{https://sites.google.com/site/zoyasvitkina/}}{zoya@google.com}{}{}
\author{Erik Vee}{Google Research \and \url{https://scholar.google.com/citations?user=1u8drP0AAAAJ&hl=en}}{erikvee@google.com}{}{}
\author{Joshua R. Wang}{Google Research \and \url{https://sites.google.com/site/joshw0}}{joshuawang@google.com}{}{}
\relatedversion{A full version of this paper is available at~\cite{fullversion}, \url{https://arxiv.org/abs/2108.02770}.}
\Copyright{Quanquan C. Liu, Manish Purohit, Zoya Svitkina, Erik Vee, Joshua R. Wang}
\authorrunning{Q. C. Liu, M. Purohit, Z. Svitkina, E. Vee, J. R. Wang}
\else
\author{Quanquan C. Liu, Manish Purohit, Zoya Svitkina, Erik Vee, Joshua R. Wang}{}{}{}
\date{}
\fi

\begin{document}
\sloppy
\raggedbottom
\maketitle

\begin{abstract}
  We consider the problem of efficiently scheduling jobs with precedence constraints on a set of
identical machines in the presence of a uniform communication delay. Such precedence-constrained jobs
can be modeled as a directed acyclic graph, $G = (V, E)$.
In this setting, if two
precedence-constrained jobs $u$ and $v$, with $v$ dependent on $u$ ($u \prec v$), are scheduled
on different machines, then $v$ must start at least $\rho$ time units after $u$ completes. The
scheduling objective is to minimize makespan, i.e.\ the total time from when the first job starts
to when the last job finishes. The focus of this paper is to provide an efficient approximation
algorithm with near-linear running time. We build on the algorithm of Lepere and Rapine [STACS
2002] for this problem to give an $O\left(\frac{\ln \rho}{\ln \ln \rho} \right)$-approximation
algorithm that runs in $\tilde{O}(|V| + |E|)$ time.
\end{abstract}

\ifsubmission
\newpage
\pagenumbering{arabic} 
\fi
\section{Introduction}
\label{sec:intro}

The problem of efficiently scheduling a set of jobs over a number of machines is a fundamental optimization problem in computer science that becomes ever more relevant as computational workloads become larger and more complex. 
Furthermore, in real-world data centers, there exists non-trivial \emph{communication delay} when data is
transferred between different machines. There is a variety 
of very recent literature devoted to the theoretical
study of this topic
~\cite{davies2020scheduling,DKRTZ21,maiti}. However, all such 
literature to date focuses on obtaining  algorithms with good approximation factors for the 
schedule length, but these algorithms require $\omega(n^2)$ time (and potentially polynomially
more) to compute the schedule. In this paper, we instead focus on efficient, 
near-linear time algorithms for scheduling while maintaining 
an approximation factor equal to that obtained by the 
best-known algorithm for our setting~\cite{LR02}.

Even simplistic formulations of the scheduling problem (e.g.\ precedence-constrained
jobs with unit length
to be scheduled on $M$ machines) are typically NP-hard,
and there is a rich body of literature on designing good approximation algorithms for the many variations of multiprocessor scheduling (refer to~\cite{brucker10} for a comprehensive history of 
such problems). Motivated by a desire to better understand the computational complexity of 
scheduling problems and to tackle rapidly growing input sizes, 
we ask the following research question:

\begin{displayquote}\emph{How computationally expensive is it to perform approximately-optimal scheduling?
}\end{displayquote}

In this paper, we focus on the classical problem of multiprocessor scheduling with communication delays on identical machines where all jobs have unit size. The jobs that need to be scheduled have data dependencies between them, where the output of one job acts as the input to another. These dependencies are represented using a directed acyclic graph (DAG) $G = (V, E)$ where each vertex $v \in V$ corresponds to a job and an edge $(u,v) \in E$ indicates that job $u$ must be scheduled before $v$. In our multiprocessor environment, if these two jobs are scheduled on different machines, then some additional time must be spent to transfer data between them. We consider the problem with \emph{uniform communication delay}; in this setting, a uniform delay of $\CommunicationDelay$ is incurred for transferring data between any two  machines. Thus for any edge $(u,v) \in E$, if the jobs $u$ and $v$ are scheduled on different machines, then $v$ must be scheduled at least $\CommunicationDelay$ units of time after $u$ finishes. Since the communication delay $\CommunicationDelay$ may be large, it may actually be more efficient for a machine to \emph{recompute} some jobs rather than wait for the results to be communicated. 
Such duplication of work can reduce schedule length by up to a logarithmic factor \cite{maiti} and 
has been shown to be effective in minimizing latency in schedulers for grid computing and cloud environments~\cite{bozdag2008compaction,casas2017balanced}. 
Our scheduling objective is to minimize the makespan of the schedule, i.e., the completion time of the last job. In the standard three field notation for scheduling problems, this problem is denoted ``$P \mid \textrm{duplication}, \textrm{prec}, p_j=1, c \mid C_{\max}$'',
\iffull
\footnote{The fields denote the following. \textbf{Identical machine information: }$P$: number, $M$, of machines is provided as input to the algorithm; \textbf{Job properties: } \textrm{duplication}: duplication is allowed; \textrm{prec}: precedence constraints; $p_j = 1$: unit size jobs; $c$: there is non-zero communication delay; \textbf{Objective:} $C_{\max}:$ minimize makespan.}
\fi
where $c$ indicates uniform communication delay.

This problem was  studied by Lepere and Rapine, who devised an $O(\ln \CommunicationDelay{} / \ln \ln \CommunicationDelay{})$-approximation algorithm for it \cite{LR02}, under the assumption that the optimal solution takes at least $\CommunicationDelay{}$ time. However, their analysis was primarily concerned with getting a good quality solution and less with optimizing the running time of their polynomial-time algorithm. A naïve implementation of their algorithm takes roughly $O(\NumberEdges{} \CommunicationDelay{} + \NumberVertices{} \ln \NumberMachines{})$ time,
where $n$ and $m$ are the numbers of vertices and edges in the DAG, respectively, and $M$ is the number of machines.
This runtime is based on two bottlenecks, (i) the computation of ancestor sets, which can be done in $O(\NumberEdges{} \CommunicationDelay{})$ time via propagating in topological order plus merging and (ii) list scheduling, which can be done in $O(\NumberVertices \ln \NumberMachines{})$ time by using a priority queue to look up the least loaded machine when scheduling a set of jobs.

However, with growing input sizes, it is highly desirable to obtain a 
scheduling algorithm whose running time is linear in the size of the 
input.
 Our primary contribution is to design a \emph{near-linear time}
 randomized approximation algorithm while preserving the 
 approximation ratio of the Lepere-Rapine algorithm:
\begin{theorem}\label{thm:1}
  There is an $O(\ln \CommunicationDelay{} / \ln \ln \CommunicationDelay{})$-approximation algorithm for scheduling jobs with precedence constraints on a set of identical machines in the presence of a uniform communication delay that runs in $O\left(n \ln M + \frac{\NumberEdges \ln^3 \NumberVertices \ln \CommunicationDelay}{\ln\ln \CommunicationDelay}\right)$ time, with high probability, assuming that the optimal solution has cost at least $\rho$.
\end{theorem}

Of course, this is tight, up to log factors, because any algorithm for this problem must respect the precedence constraints, which require $\Omega(\NumberVertices + \NumberEdges)$ time to read in.
In the settings where our algorithm is more efficient than Lepere-Rapine, the approximation factor
of the algorithm is still very small (near-constant in the cases where $\rho = \poly \log n$), yet our algorithm achieves
a better runtime while maintaining the same approximation compared to the previous best-known algorithm for the problem. 

\subsection{Related Work}
Algorithms for scheduling problems under different models have been studied for decades, and there is a rich literature on the topic (refer to~\cite{brucker10} for a comprehensive look). Here we review work on theoretical aspects of scheduling with communication delay, which is most relevant to our results.

Without duplication, scheduling a DAG of unit-length jobs with unit communication delay was shown to be NP-hard by Rayward-Smith \cite{rayward1987uet}, who also gave a 3-approximation for this problem. Munier and K\"{o}nig gave a $4/3$-approximation for an unbounded number of machines \cite{munier+k:schedule}, and Hanen and Munier gave a $7/3$-approximation for a bounded number of machines \cite{hanen+m}. Hardness of approximation results were shown in \cite{bampis+gk:schedule,hoogeveen+lv:schedule,picouleau}. In recent results, Kulkarni et al.\ \cite{kulkarni2020soda} gave a quasi-polynomial time approximation scheme for a constant number of machines and a constant communication delay, whereas Davies et al.\ \cite{davies2020scheduling} gave an $O(\log \rho \log M)$ approximation for general delay and number of machines. 
Even more recently, Davies et al.\ \cite{DKRTZ21} presented a $O(\log^4 n)$-approximation algorithm for
the problem of minimizing the weighted sum of completion times on \emph{related machines} in the presence of
communication delays. They also obtained a $O(\log^3 n)$-approximation algorithm under the same model
but for the problem of minimizing makespan
under communication delay. Notably, \emph{none} of the aforementioned algorithms consider duplication and the most recent
algorithms have running times that are large polynomials.

Allowing the duplication of jobs was first studied by Papadimitriou and Yannakakis \cite{papadimitriou1990towards}, who obtained a 2-approximation algorithm for scheduling a DAG of identical jobs on an unlimited number of identical machines. A number of papers have improved the results for this setting \cite{ahmad+k:schedule,darbha+a:schedule,palis1996task}. With a finite number of machines, Munier and Hanen \cite{munier+h:schedule} proposed a 2-approximation algorithm for the case of unit communication delay, and Munier \cite{munier1999approximation} gave a constant approximation for the case of tree precedence graphs. For a general DAG and a fixed delay $\rho$, Lepere and Rapine \cite{LR02} gave an algorithm that finds a solution of cost $O(\log \rho / \log \log \rho) \cdot (OPT + \rho)$, which is a true approximation if one assumes that $OPT \geq \rho$. This is the main result that our paper builds on. It applies to a set of identical machines and a set of jobs with unit processing times. Recently, an $O(\log M \log \rho / \log \log \rho)$ approximation has been obtained for a more general setting of $M$ machines that run at different speeds and jobs of different lengths by Maiti et al.\ \cite{maiti}, also under the assumption that $OPT \geq \rho$.
However, the running time of this algorithm is a large polynomial ($\omega(n^2)$),
as it requires solving an LP with $\Omega(Mn^2)$ variables.

Our results so far only apply to scheduling with duplication. In Maiti et al.\ \cite{maiti}, a polynomial-time reduction is presented that transforms a schedule with duplication into one without duplication (with a polylogarithmic increase in makespan). However, this reduction involves constructing an auxiliary graph of possibly $\Omega(\rho^2)$ size, and thus does not lend itself easily to a near-linear time algorithm. It would be interesting to see if a near-linear time reduction could be found.

\subsection{Technical Contributions}

A na\"{i}ve implementation of the Lepere-Rapine algorithm is bottlenecked by the need to determine the set of all ancestors of a vertex $v$ in the graph, as well as the intersection of this set with a set of already scheduled vertices. Since the ancestor sets may significantly overlap with each other, trying to compute them explicitly (e.g., using DFS to write them down)  results in superlinear work. We use a variety of technical ideas to only compute the essential size information that the algorithm needs to make decisions about these ancestor sets.

\begin{itemize}
  \item \textbf{Size estimation via sketching.} We use streaming techniques to quickly estimate the sizes of all ancestor sets simultaneously. It costs $O((|V| + |E|)\log^2 n)$ time
  to make such an estimate once, so we are careful to do so sparingly.
  \item \textbf{Work charging argument.} Since we cannot compute our size estimates too often, we still need to perform some DFS for ancestor sets. We control the amount of work spent doing so by carefully charging the edges searched to the edges we manage to schedule.

  \item \textbf{Sampling and pruning.} 
  Because we cannot brute-force search all ancestor sets, 
  we randomly sample vertices,
  using a consecutive run of unscheduleable vertices
  as evidence that many vertices are not schedulable.
  This allows us to pay for an expensive size-estimator computation to prune many ancestor sets simultaneously.
\end{itemize}

\subsection{Organization} The main contribution of this paper is our 
algorithm for scheduling small subgraphs in near-linear time. 
We provide a detailed description and analysis
of this algorithm  in~\cref{sec:detailed-alg}.
Then, we proceed with our algorithm for scheduling general graphs in
\cref{sec:general}.

\section{Problem Definition and Preliminaries}
\label{sec:prelims}

An instance of scheduling with communication delay is specified by a directed acyclic graph $G = (V,E)$, a quantity $M \ge 1$ of identical machines, 
and an integer communication delay $\rho > 1$. We assume that time is slotted and let $T = \{1, 2, \ldots\}$ denote the set of integer times. Each vertex $v \in V$ corresponds to a job with processing time 1 and a directed edge $(u, v) \in E$ represents the precedence constraint that job $v$ depends on job $u$. In total, there are $\NumberVertices{} = |V|$ vertices (representing jobs) and $\NumberEdges{} = |E|$ precedence constraints. 
The parameter $\CommunicationDelay{}$ indicates the amount of time required to communicate the result of a job computed on one machine to another. In other words, a job $v$ can be scheduled on a machine at time $t$ only if all jobs $u$ with $(u, v) \in E$ have either completed on the same machine before time $t$ or on another machine before time $t - \CommunicationDelay{}$. We allow for a job to be \emph{duplicated}, i.e., copies of the same vertex $v \in V$ may be processed on different machines. Let $\mathcal{M}$
be the set of machines available to schedule the jobs. A schedule $\sigma$ is represented by a set of triples $\{(m, v, t)\} \subset \mathcal{M} \times V \times T$ where each triple represents that job $v$ is scheduled on machine $m$ at time $t$.
The goal is to obtain a feasible schedule that minimizes the makespan, i.e., the completion time of the last job. Let $\opt$ denote the makespan of an optimal schedule. Since $\CommunicationDelay{}$ represents the amount of time required to communicate between machines, and in practice, any schedule must communicate the results of the computation, we assume that $\opt \ge \CommunicationDelay$ as is standard in literature~\cite{LR02, maiti}.

\begin{table}[htb]
    \centering
    \begin{tabular}{c l}
    \toprule
         Symbol & Meaning \\
         \midrule
         $G = (V, E)$ & main input graph  \\
         $\NumberVertices = |V|, \NumberEdges{} = |E|$ & number of vertices / edges \\ 
         $H = (V_H, E_H)$ & subgraph to be scheduled in each phase\\
         $\CommunicationDelay$ & communication delay \\
         $u, v$ & vertices \\
         $\AncestorSetH{v}{H}$ & set of ancestors of vertex $v$ in graph $H$ including $v$\\
         $\AncestorSetH{S}{H}$ & $\AncestorSetH{S}{H} = \bigcup_{v \in S} \AncestorSetH{v}{H}$ in graph $H$\\
         $\sedgeh{v}{H}$ [resp., $\sedgeh{S}{H}]$ & edges induced by $\AncestorSetH{v}{H}$ [resp., $\AncestorSetH{S}{H}$] in graph $H$\\
         $\estvh{v}{H}, \esteh{v}{H}$ & estimated size of $\svh{v}{H}$ and $\sedgeh{v}{H}$\\ 
         $\NumberMachines{}$ & number of machines\\
         $\gamma$ & threshold for \good vs. \bad vertices\\
         \bottomrule
    \end{tabular}
    \caption{Table of Symbols}
    \label{tab:symbol_table}
\end{table} 

We now set up some notation to help us better discuss dependencies arising from the precedence constraints of $G$. For any vertex $v \in V$, let $\pred(v) \triangleq \{u \in V \mid (u,v) \in E\}$ be the set of (immediate) predecessors of $v$ in the graph $G$, and similarly let $\successor(v) \triangleq \{w \in V \mid (v,w) \in E\}$ be the set of (immediate) successors.
For $H = (V_H, E_H)$, a subgraph of $G$, we use $\AncestorSetH{v}{H} \triangleq \{u \in V_H \mid \exists \text{ a directed path from $u$ to $v$ in $H$}\} \cup \{v\}$ to denote the set of (indirect) ancestors of $v$, including $v$ itself. Similarly, for $S\subseteq V$, we use $\AncestorSetH{S}{H} \triangleq \bigcup_{v \in S} \AncestorSetH{v}{H}$ to denote the indirect ancestors of the entire set $S$. 
We use $\sedgeh{S}{H}$ to denote the edges of the subgraph induced by $\AncestorSetH{S}{H}$. We drop the subscript $H$ when the subgraph $H$ is clear from context.
Throughout, we use the phrase \emph{with high probability} 
to indicate with probability at least $1 - \frac{1}{n^c}$ for any constant $c \geq 1$.

For convenience, we summarize the notation we
use throughout the paper in \cref{tab:symbol_table}.

\section{Technical Overview}\label{sec:overview}
We start by reviewing the algorithm of Lepere and Rapine \cite{LR02}, shown in Algorithm \ref{alg:lr}, as our algorithm follows a similar outline. Then we describe the technical improvements
of our algorithm to achieve near-linear running time.

\begin{algorithm}[hbt]
\caption{Outline of Lepere Rapine Scheduling Algorithm  \cite{LR02}}
\label{alg:lr}
\LinesNumbered
\While{$G$ is non-empty} {\label{line:outer-loop}
    Let $H$ be a subgraph of $G$ induced by vertices with at most $\rho + 1$ ancestors\\
    \While{$H$ is non-empty}{\label{line:inner-loop}
        \For{each vertex $v$ in $H$}{
            \If{greater than $\gamma$ fraction of $\AncestorSetH{v}{H}$ is unscheduled}{
                Add $\AncestorSetH{v}{H}$, in topological order, to a machine with earliest end time
            }
        }
        Insert a delay until $C+\rho$ on all machines, where $C$ is the latest end time\\
        Remove scheduled vertices from $H$
    }
    Delete vertices in $H$ from $G$
}
\end{algorithm}

\begin{figure}[htb]
     \centering
     \begin{subfigure}[b]{0.35\textwidth}
         \centering
         \includegraphics[width=\textwidth]{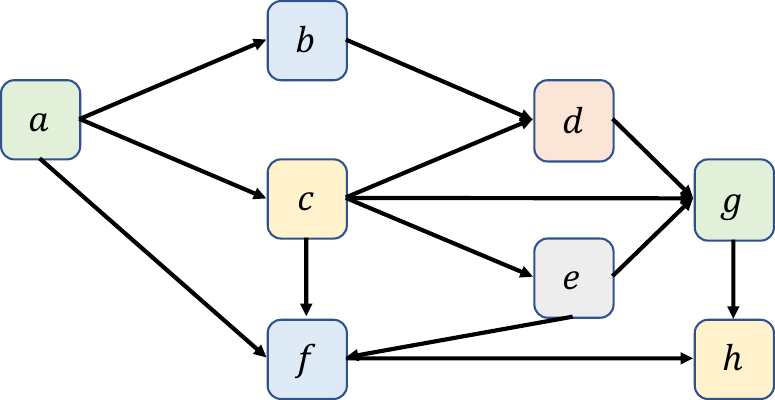}
         \caption{Initial input DAG.}
         \label{fig:initial-graph}
     \end{subfigure}
     \hfill
     \begin{subfigure}[b]{0.55\textwidth}
         \centering
         \includegraphics[width=\textwidth]{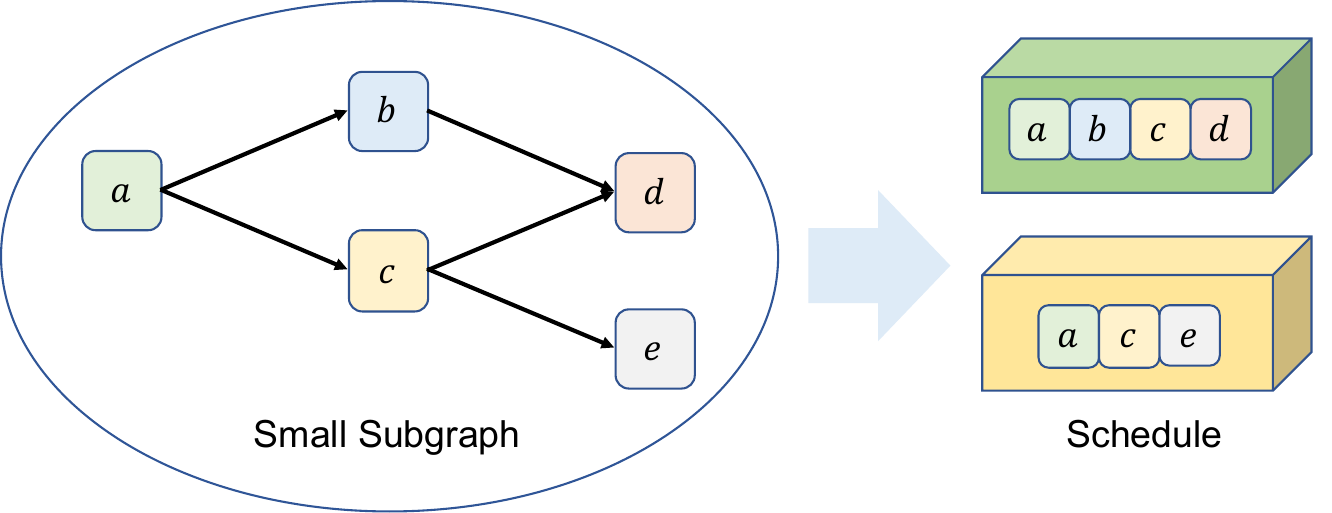}
         \caption{Find a small subgraph in the input graph and schedule the small subgraph.}
         \label{fig:schedule-small}
     \end{subfigure}
     \par\bigskip
     \begin{subfigure}[b]{0.6\textwidth}
         \centering
         \includegraphics[width=\textwidth]{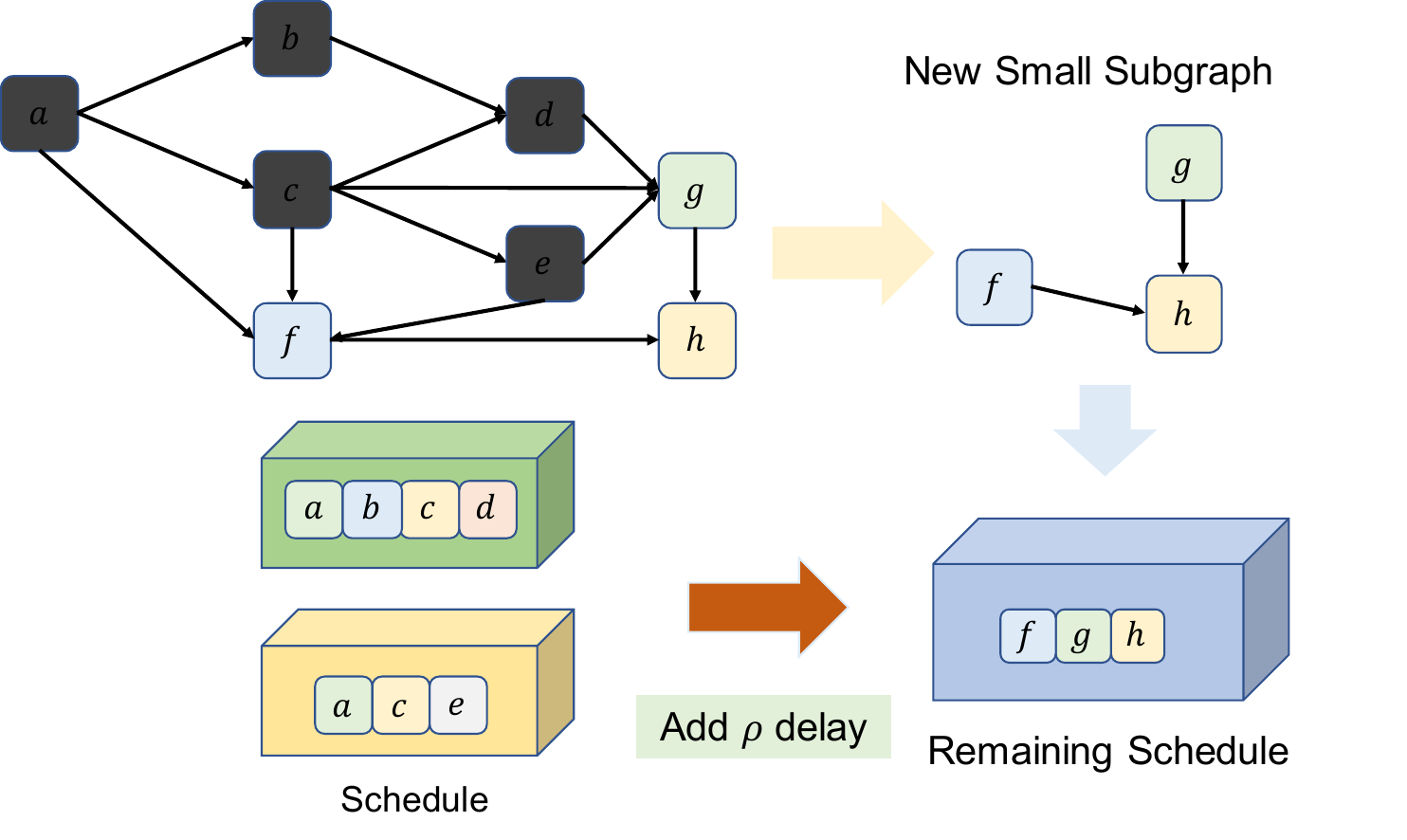}
         \caption{Remove scheduled vertices from the graph. Add a $\rho$ communication delay to the schedule after the previously scheduled jobs.
         Find a new small subgraph and
         schedule it.}
         \label{fig:removed-scheduled}
     \end{subfigure}
    \caption{Overview of the Lepere-Rapine algorithm for scheduling general graphs.}
    \label{fig:main-alg-overview}
\end{figure}

\myparagraph{Description of Lepere-Rapine~\cite{LR02}}
The outer loop (\cref{line:outer-loop}) iteratively finds \emph{small subgraphs} of $G$ which consist of vertices that have \emph{height} at most $\rho + 1$. We show in this paper that instead of considering their definition of \emph{height}, it is sufficient in our
algorithm
to consider small subgraphs to be those with at most $2\rho$ ancestors. We call one iteration of this loop a \emph{phase}. Within the phase, $H$ is fully scheduled, after which the algorithm goes on to the next ``slice'' of $G$. However, $H$ is not scheduled all at once, but instead each iteration of the inner \texttt{while} loop (\cref{line:inner-loop}) schedules a subset of $H$, which we call a \emph{batch}. To determine which vertices of $H$ make it into a batch, the algorithm checks the fraction of ancestors of each vertex that have already been scheduled in the same batch.
If this fraction for a vertex $v$ is low (we call $v$ \emph{\good} in that case), then its ancestor set $\AncestorSetH{v}{H}$ is list-scheduled as a unit,
i.e.\ ancestor jobs are duplicated, topologically sorted, and placed on one machine. If the fraction of scheduled ancestors is high (in which case we call $v$ \emph{\bad}), $v$ is skipped in this iteration.
We skip $v$ to avoid excessive duplication
that would create too much load on the machines.
After each batch is placed on the machines, a delay of $\rho$ is added to the end of the schedule to allow all the results to propagate. This allows the scheduled jobs to be deleted from $H$.
This algorithm is illustrated pictorially in~\cref{fig:main-alg-overview}.

\myparagraph{Runtime Challenges with Lepere-Rapine}
Naively, both finding the small subgraphs as well as determining each batch takes $\Omega(n\rho^2)$ time.
Determining which nodes belong in the current small subgraph is a matter of whether their ancestor counts are more than $\rho$ or at most $\rho$.
A standard procedure would be to apply DFS and merge ancestor sets, but that can easily run in $\Omega(\rho^2)$ time per node (a node may have $\Omega(\rho)$ direct parents, each with an ancestor size of $\Omega(\rho)$ that needs to get merged in).

The other technical hurdle is in determining the batches to schedule. We would like to schedule vertices whose ancestors do not \emph{overlap too much}. To illustrate the difficulty of applying sketching-based methods (e.g. min-hash), consider the following example. Suppose that $\rho^2$ elements have already been scheduled in this batch. Now, we want to find the number of ancestors of vertex $v$,
$\AncestorSet{v}$, that intersect
with the currently scheduled batch, where $|\AncestorSet{v}| \leq \rho$
by construction. By the lower bound given in~\cite{PSW14}, even estimating (up to $1 \pm \epsilon$ relative error with constant probability) the size of this intersection would require sketches of size at least
$\eps^{-2}(\rho^2/\rho) = \eps^{-2}\rho$. Using such $\rho$-sized sketches over all batches and all small subgraphs require $\Omega(n\rho)$ time in total.

Since $\rho$ may be super-logarithmic, these naive implementations don't quite meet our goal of a near-linear time algorithm. To summarize, the two main technical challenges for our setting are the following:

\begin{challenge}\label{chal:schedule}
    We must be able to find the small subgraphs in near-linear time.
\end{challenge}

\begin{challenge}\label{chal:batch}
    We must be able to find the vertices to add to each batch $\batch$ in near-linear time.
\end{challenge}

We solve~\cref{chal:schedule} by relaxing the definition of small subgraph and using \emph{count-distinct} estimators (discussed in~\cref{sec:estimators}).
The majority of our paper focuses on solving~\cref{chal:batch} which
requires several new techniques for the problem
outlined in the rest of this section (\cref{sec:sampling} and \cref{sec:pruning}).
The below procedures run on a small subgraph,
$H = (V_H, E_H)$, where
the number of ancestors of each vertex is bounded by $2\rho$. Note the factor of $2$ results from
our count-distinct estimator. This is described in~\cref{sec:detailed-alg}. Our algorithm for scheduling
small subgraphs is shown pictorially in~\cref{fig:alg-overview}.

\subsection{Sampling Vertices to Add to the Batch}\label{sec:sampling}

We first partition the set of unscheduled vertices in $V_H$ into buckets based on the estimated number of edges in the subgraph induced by their ancestors. (We place vertex $v$---if it has no
ancestors---into the bucket with the smallest index.)
We partition by edges instead of vertices because the number of edges
in the induced subgraph of the ancestors affects our running time.
More formally, let $S_i$ be the set of vertices not yet scheduled in iteration $i$ (\cref{line:inner-loop}, \cref{alg:lr}). We partition $S_i$ into $k = O(\log \rho)$ buckets $K_1, \ldots, K_k$ such that bucket $K_j$ contains all vertices $w \in S_i$ where $\este{w} \in [2^j, 2^{j+1})$; $\este{w}$ denotes the estimated number of edges in the subgraph induced by ancestors of $w$.

From each bucket $\bucket_j$, in decreasing order of $j$, we  sample vertices, sequentially,
without replacement.
For each sampled vertex $v$,
we enumerate its ancestors
and determine how many are in the current batch $\batch$.
If at least a $\gamma$-fraction of
the vertices \emph{are not in $\batch$ and} at least a $\gamma$-fraction of
the edges (with both endpoints in $B$)
in the induced subgraph $\sgraphh{v}{H_i}$ are
\emph{not in} $\batch$,
then add $v$ and all its ancestors to
$\batch$. We call such a vertex $v$ \defn{\good}.
Otherwise,
we do not add $v$
to $\batch$ and
label this vertex as
\defn{\bad}.
For our algorithms, we set $\gamma = \frac{1}{\sqrt{\rho}}$ to
minimize the approximation factor but
$\gamma$ can be set to any value
$\gamma < 1/2$.
Lepere-Rapine did not consider edges in their algorithm because
the number of edges in the induced subgraph does not
affect the schedule length; however,
considering edges is crucial for our algorithm to run in near-linear time.

For each bucket sequentially, we sample vertices
uniformly at random,
until we have sampled
$O(\log n)$ consecutive vertices
that are \emph{\bad}
(or we have run
out of vertices and the bucket is empty).
Then, the key intuition is that for every $v$ that we add to
$\batch$, we can afford to \emph{charge the cost of
enumerating the ancestor set} for $O(\log n)$ additional vertices
in the same bucket as well as $O(\log n)$ additional vertices in each bucket with smaller $j$ to it. Because we are looking at
buckets
with decreasing indices, we can charge the
additional vertices found in future buckets
to the most recently found \good vertex.
\iffull
To see why we can charge the samples from buckets with smaller $i$,
suppose that one vertex $v$ in bucket $i$ was added to $B$ and no vertices
in buckets $(i, \log \rho]$ were added to $B$. Then, the cost charged to
$v$ of enumerating the $O(\log n)$ ancestor sets in buckets $[i, \log \rho]$
is at most $\sum_{j = i}^{\log \rho} (2^{\log \rho - j}) \log n =
O(2^{\log \rho - i} \log n)$, asymptotically
the same cost as charging the sampled vertices from bucket $i$.
\fi

\subsection{Pruning All \Bad Vertices from Buckets}\label{sec:pruning}
After we have performed the sampling procedure, we are still
not done. Our goal is to make sure that
\emph{all vertices which are not included} in $\batch$ are approximately \bad.
This means that we must remove the \bad
vertices so
that we can perform our sampling procedure
again in a smaller
sample space in order to find additional
\good vertices.
To accomplish this, we perform a \defn{pruning} procedure
involving re-estimating the ancestor sets consisting of vertices
that have not been added to the batch.
Using these estimates, we remove \emph{all} \bad
vertices
from our buckets. Note that we
\emph{do not rebucket the vertices}
because none
of the ancestor sets of the vertices changed
sizes.
Then, we perform
our sampling procedure above (again) to find more \good vertices. The key is that since
we removed all \bad vertices, \emph{the first
sampled vertex from the non-empty bucket with the largest index
is \good}.
\\\\
We perform the above sampling and pruning
procedures until each
bucket is empty.
Then, we schedule the batch
and remove all scheduled vertices from
$H$ and proceed again with the procedure until
the graph is empty.
We perform a standard simple
greedy list scheduling algorithm (\cref{app:list-scheduling})
on our batch
on $M$ machines.
\vspace{1em}
\begin{figure}[htb!]
     \centering
     \begin{subfigure}[b]{0.35\textwidth}
         \centering
         \includegraphics[width=\textwidth]{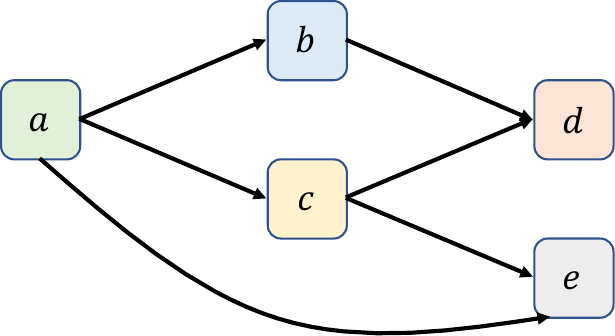}
         \caption{Input small subgraph.}
         \label{fig:initial-small-graph}
     \end{subfigure}
     \hfill
     \begin{subfigure}[b]{0.4\textwidth}
         \centering
         \includegraphics[width=\textwidth]{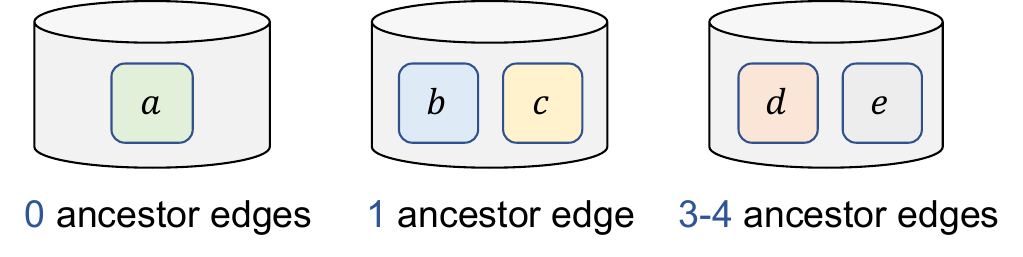}
         \caption{Vertices are bucketed according to the estimate of the number of edges in the induced subgraph of its ancestors.}
         \label{fig:bucketing}
     \end{subfigure}
     \par\bigskip
     \begin{subfigure}[b]{0.7\textwidth}
         \centering
         \includegraphics[width=\textwidth]{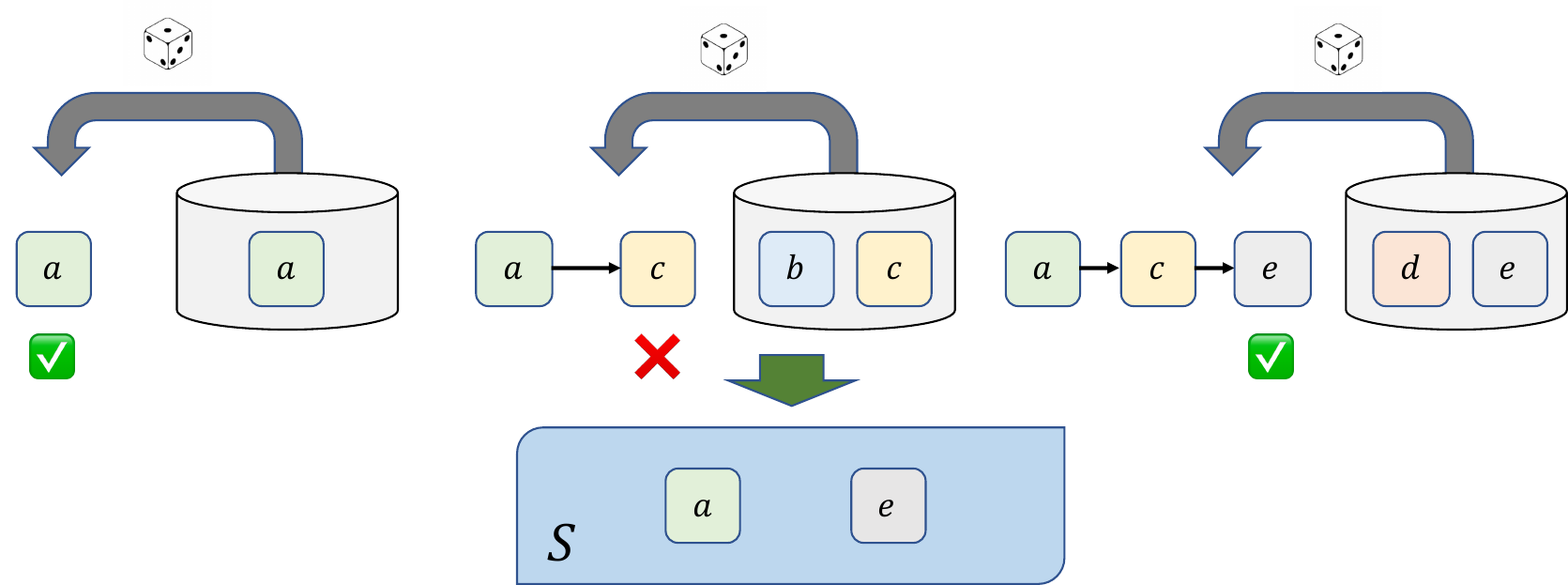}
         \caption{Vertices are uniformly at random sampled from buckets. Then, vertices which have sufficiently many ancestors and ancestor edges not in $S$ are added to $S$.}
         \label{fig:sampling}
     \end{subfigure}
     \par\bigskip
    \begin{subfigure}[b]{0.4\textwidth}
         \centering
         \includegraphics[width=\textwidth]{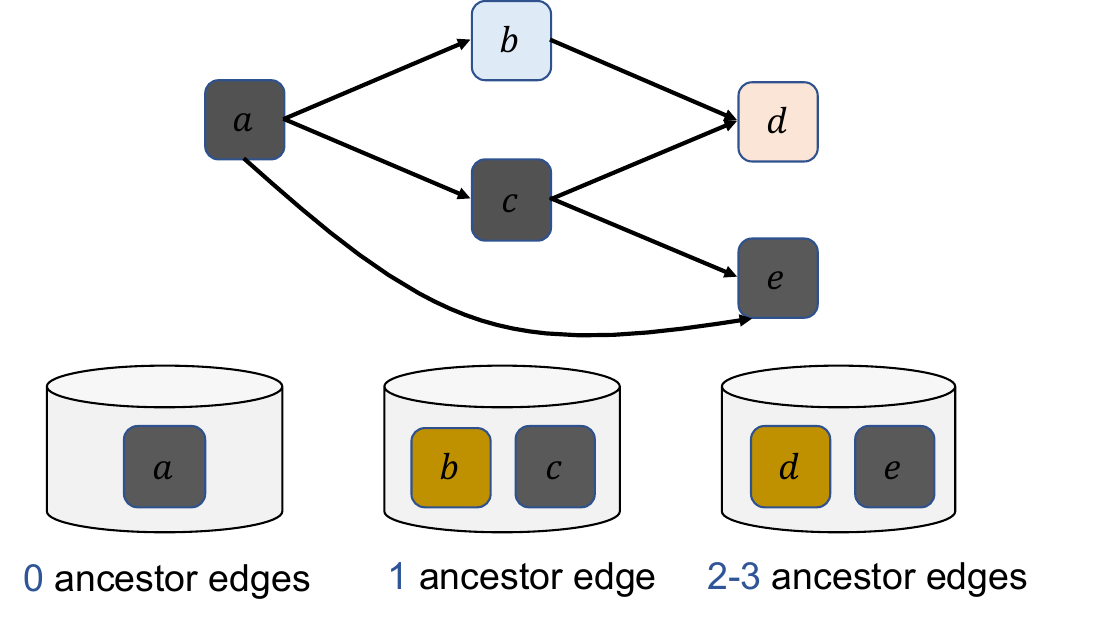}
         \caption{Vertices which are in $S$ or have a large proportion of ancestors or ancestor edges in $S$ are pruned from buckets. $b$ and $d$ are pruned in this example.}
         \label{fig:pruning}
     \end{subfigure}
     \hfill
    \begin{subfigure}[b]{0.3\textwidth}
         \centering
         \includegraphics[width=0.6\textwidth]{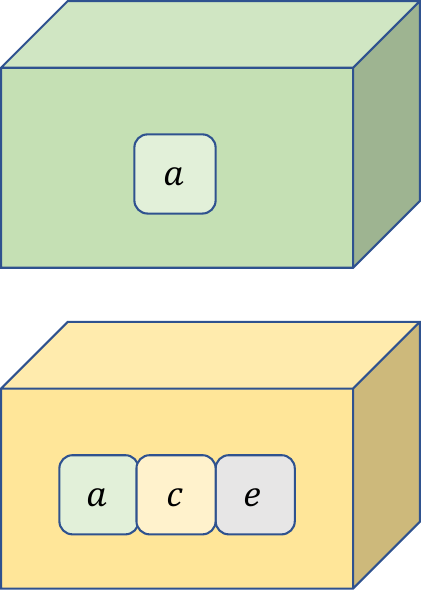}
         \caption{Vertices in $S$ and all ancestors are scheduled by duplicating ancestors and list scheduling.}
         \label{fig:list-schedule}
     \end{subfigure}
        \caption{Overview of our scheduling small subgraphs algorithm. We choose $\gamma = 2/3$ here for illustration purposes but in our algorithms $\gamma < 1/2$.}
        \label{fig:alg-overview}
\end{figure}

\iffull
\fi

\section{Estimating Number of Ancestors}
\label{sec:estimators}

Let $\tilde{G} = (\tilde{V}, \tilde{E})$ be an arbitrary directed, acyclic graph. We first present our algorithm to estimate the number of ancestors of any vertex $v \in \tilde{V}$. Consider any vertex $v \in \tilde{V}$ and let $p_1, p_2, \ldots p_\ell$ be the predecessors of $v$ in $\tilde{G}$.
Then we have $\AncestorSetH{v}{\tilde{G}} = \cup_{i=1}^\ell \AncestorSetH{p_i}{\tilde{G}} \cup \{v\}$ and hence
$|\AncestorSetH{v}{\tilde{G}}|$ is the number of distinct
elements in the multiset $\cup_{i=1}^\ell
\AncestorSetH{p_i}{\tilde{G}} \cup \{v\}$.
In order to estimate
$|\AncestorSetH{v}{\tilde{G}}|$ efficiently, we use a
procedure to estimate the number of distinct elements in a
data stream. This problem,  known as the \emph{count-distinct problem}, is well studied and many efficient estimators exist
\cite{AMS96,BJKST02,WVT90,FFGM12,KNW10}. Since we need to
estimate $|\AncestorSetH{v}{\tilde{G}}|$ for all vertices $v
\in \tilde{V}$ in near-linear time, we require an additional
\emph{\mergeable} property to ensure that we can efficiently
obtain an estimate for $|\AncestorSetH{v}{\tilde{G}}|$ from
the estimates of the parent ancestor set sizes
$\{|\AncestorSetH{p_1}{\tilde{G}}|, \ldots,
|\AncestorSetH{p_\ell}{\tilde{G}}|\}$.

We formally define the notion of a
\emph{count-distinct estimator} and the \mergeable property.

\begin{definition}\label{def:distinct-count}
For any multiset $\cS$, let $|\cS|$ denote the number of distinct elements in $\cS$. We say $T$ is an $(\eps, \delta, D)$-$\est$ for $\cS$ if it uses space $D$ and returns a value $\hat{s}$ such that $(1 - \eps)|\cS| \leq \hat{s} \leq (1 + \eps)|\cS|$ with probability at least $(1-\delta)$.
\end{definition}

\begin{definition}[\Mergeable Property]\label{def:anc-estimator}
An $(\eps, \delta, D)$-$\est$ exhibits the \emph{\mergeable property} if estimator
$T_1$ for multiset $\mathcal{S}_1$ and estimator $T_2$ for multiset $\mathcal{S}_2$ can be merged
in $O(D)$ time using $O(D)$ space into an $(\eps, \delta, D)$-estimator for $\mathcal{S}_1 \cup
\mathcal{S}_2$.
\end{definition}

We note that the \emph{count-distinct estimator} in \cite{BJKST02} satisfies the \mergeable property and suffices for our purposes. We include a description of the procedure and a proof of the \mergeable property in \cref{sec:estimator}.

\begin{lemma}[\cite{BJKST02}]\label{lem:estimator-bound}
For any constant $\eps > 0$ and $d \geq 1$, there exists an $\left(\eps, \frac{1}{n^d}, O\left(\frac{1}{\eps^2}\log^2 n\right)\right)$-$\est$
that satisfies the mergeable property where $n$ denotes an upper bound on the number of distinct elements.
\end{lemma}

Given such an estimator, one can readily estimate the number of ancestors of each vertex $v \in \tilde{V}$ in near-linear time by traversing the vertices of the graph in topological order. An estimator for vertex $v$ can be obtained by \emph{merging} the estimators for each predecessor of $v$.
Similarly, we can also estimate the number of edges $|\sedge{v}|$ in the subgraph induced by ancestors of any vertex $v$ in near-linear time. We defer a detailed description of these procedures to \cref{alg:ancestor-estimation}
\ifsubmission
in~\cref{sec:estimator} and Algorithm 8  in our
full paper~\cite{fullversion}.
\else
and \cref{alg:ancestor-edge-estimation} in \cref{sec:estimator}.
\fi

\begin{lemma}\label{lem:ancestor-est-correct}
Given any input graph $\tilde{G} = (\tilde{V}, \tilde{E})$ and constants $\eps > 0, d \geq 1$, there exists an algorithm that runs in $O\left((|\tilde{V}| + |\tilde{E}|)\log^2 n\right)$ time and returns estimates $\estv{v}$ and $\este{v}$ for each $v \in \tilde{V}$ such that $(1-\eps)|\AncestorSetH{v}{\tilde{G}}| \leq \estv{v} \leq (1+\eps)|\AncestorSetH{v}{\tilde{G}}|$ \emph{and} $(1-\eps)|\sedgeh{v}{\tilde{G}}| \leq \este{v} \leq (1+\eps)|\sedgeh{v}{\tilde{G}}|$ with probability at least $1 - \frac{1}{n^d}$.
\end{lemma}

\begin{proof}
\cref{lem:desired-estimator} provides us with our
desired approximation. Now, all that remains to show
is that~\cref{alg:ancestor-estimation}
\ifsubmission
and Algorithm 8  in our full paper~\cite{fullversion}
\else
and~\cref{alg:ancestor-edge-estimation}
\fi
runs within
our desired time bounds. \cref{alg:ancestor-estimation}
visits each
vertex exactly once. For each vertex, it merges
the estimators of each of its immediate predecessors.
By~\cref{lem:merge-time}, each merge takes $O\left(\frac{1}{\eps^2}\log^2 n\right)$ time. Because
we visit each vertex exactly once, we also visit each
predecessor edge exactly once. This means that in
total we perform $O\left(\frac{m}{\eps^2}\log^2 n\right)$
merges. Since $\eps$ is constant, this algorithm
requires $O(m\log^2 n)$ time. The same proof follows
\ifsubmission
for Algorithm 8
in our full paper~\cite{fullversion}.
\else
for~\cref{alg:ancestor-edge-estimation}.
\fi
\end{proof}

Throughout the remaining parts of
the paper, we assume that $\eps =
1/3$ in our estimation
procedures and do not
explicitly give our results in
terms of $\eps$.

\section{Scheduling Small Subgraphs in Near-Linear Time}\label{sec:detailed-alg}

Here, we consider subgraphs $H = (V, E)$ such that every vertex in the graph has a bounded number of ancestors and obtain a schedule for such \emph{small subgraphs} in near-linear time.

\begin{definition}\label{def:small-subgraph}
A \emph{small subgraph} is a graph $H = (V_H, E_H)$ where
each vertex $v \in V_H$ has at most
$2\rho$ ancestors.
\end{definition}

Our main algorithm schedules a small graph in batches using~\cref{alg:core_algo}.
After scheduling a batch of vertices, we insert a communication delay of
$\rho$ time units so that results of the computation from the previous batch
are shared with all machines (similar to Lepere-Rapine).
Then, we remove all vertices that we
scheduled and compute the next batch from the smaller graph.
We present this algorithm in~\cref{alg:small-subgraph}.

\begin{algorithm}[hbpt]
\caption{ScheduleSmallSubgraph$(H, \gamma)$}\label{alg:small-subgraph}
\LinesNumbered
\SetAlgoLined
\KwResult{A schedule of small subgraph
 $H$ on $M$ processors.}
 \KwIn{$H = (V_H, E_H)$ where $|\svh{v}{H}| \leq 2\rho$
 for all $v \in V_H$ and parameter $0 < \gamma < 1/2$.}
 \While{$H \neq \emptyset$}{\label{line:small-subgraph-while}
    $\batch \leftarrow$ FindBatch($H, \gamma$). [\cref{alg:core_algo}] \label{line:num-batches}\\
    List schedule $\as{v}$ using the $M$ processors
    for all $v\in \batch$. (\cref{app:list-scheduling})\\
    Insert communication delay of $\rho$ time units into the schedule.\label{line:comm-delay}\\
    Remove each $v \in \sv{\batch}$ and all edges adjacent to $v$ from $H$.\\
 }
\end{algorithm}

Our algorithm for scheduling small subgraphs relies on two key building blocks -- estimating the sizes of the ancestor sets (and ancestor edges) of each vertex (\cref{sec:estimators}), and using these estimates to find a \emph{batch} of vertices that can be scheduled without any communication (possibly by duplicating some vertices). We show how to find a batch in \cref{sec:batching}.

\subsection{Batching Algorithm}
\label{sec:batching}

Recall that the plan is for our algorithm to schedule a small subgraph by successively scheduling maximal subsets of vertices in the
graph whose ancestors do not 
\emph{overlap too much}; we call such a set of vertices a \defn{batch}. After scheduling each batch, we remove all the scheduled vertices from the graph and iterate on the remaining subgraph. 

A detailed description of this procedure is given in \cref{alg:core_algo}. 
For each vertex $v \in V_H$, let $\estv{v}$ and $\este{v}$ denote the estimated sizes of $\AncestorSetH{v}{H}$ and $\sedgeh{v}{H}$ respectively (henceforth referred to as
$\AncestorSet{v}$ and $\sedge{v}$). Then the $i$-th bucket, $\bucket_i$, is defined as $\bucket_i = \{v \in V_H \mid 2^i \leq \este{v} < 2^{i+1} \}$. Since every node $v \in V_H$ has at most $O(\rho)$ ancestors, there are only $k = O(\log \rho)$ such buckets. Recall that from \cref{lem:ancestor-est-correct}, this estimation can be 
performed in near-linear time. The algorithm maintains a batch $\batch$ of vertices that is initially empty. For each non-empty bucket $\bucket_i$ (processed in decreasing order of size), we repeatedly sample nodes uniformly at random from the bucket (without replacement). 

For each sampled node $v \in \bucket_i$, we explicitly enumerate the ancestor sets $\AncestorSet{v}$ and $\es{v}$ and also compute $\AncestorSet{v} \setminus \AncestorSet{B}$ and $\es{v} \setminus \es{B}$. Since we can maintain the ancestor sets of the current batch $\batch$ in a hash table, this enumeration takes $O(|\es{v}|)$ time. 
A sampled node $v$ is said to be \emph{\good} if $|\AncestorSet{v} \setminus \AncestorSet{B}| > \gamma |\AncestorSet{v}|$ \emph{and} $|\es{v} \setminus \es{B}| > \gamma |\es{v}|$; and said to be \emph{\bad} otherwise. The algorithm adds all \good nodes to the batch $\batch$ and continues sampling from the bucket until it samples $\Theta(\log n)$ consecutive \bad nodes. 
Once all the buckets have been processed, we \emph{prune} the buckets to remove all stale nodes and then repeat the sampling procedure until all buckets are empty. 

A bucket $\bucket_i$ is reduced when 1) a vertex, $v$, in it
is added to $\batch$, 2) a sampled vertex $v$ is stale and 3) during the pruning process. No vertex remains unscheduled because either a vertex 
$v$ is scheduled in the current batch or it is stale. For all stale vertices, in~\cref{alg:small-subgraph}, 
we remove all the scheduled ancestors of these stale vertices (so the 
vertices become fresh again). We repeat the procedure given 
in~\cref{alg:core_algo} (\cref{line:num-batches} 
of~\cref{alg:small-subgraph}) until the entire graph is scheduled (\cref{line:small-subgraph-while} of~\cref{alg:small-subgraph}) 
so that all vertices are eventually scheduled.
The pruning procedure is presented in \cref{alg:pruning}. 
In this step, we again estimate the sizes of ancestor sets of all vertices in the graph $H \setminus \as{B}$ to determine whether a vertex is stale.

\begin{algorithm}[hbtp]
\caption{FindBatch($H, \gamma$) 
}\label{alg:core_algo}
\KwResult{Returns batch $\batch$, the batch of vertices to schedule.}
\KwIn{A subgraph $H = (V_H, E_H)$ such that $|\AncestorSetH{v}{H}| \leq 2\rho$ for all $v \in V_H$; $0 < \gamma < 1/2$.} 
Let $N = \Theta(\log n)$.\\
Initially, $\batch \leftarrow \emptyset$ and all nodes are unmarked.\\
Obtain estimates $\estv{v}$ and $\este{v}$ for all $v \in V_H$.\\
Let bucket $\bucket_i = \{ v \in V_H: 2^i \leq \este{v} < 2^{i+1}\}$.\\
\While{at least one bucket is non-empty}{ \label{line:while}
\For{$i = k \text{ to } 1$}{\label{line:forloop}
    Let $s = 0$.\\
    \While{$s < N$ and $|\bucket_i|>0$}{
        Let $v$ be a uniformly sampled node in bucket $\bucket_i$.\\
        Find $\as{v}$ and $\as{v}\setminus \as{\batch}$ as well as $\es{v}$ and $\es{v}\setminus\es{\batch}$.\label{line:find-anc} \label{line:ancestor}\\
        \eIf{ $|\as{v}\setminus \as{\batch}| > \gamma|\as{v}|$ and $|\es{v}\setminus \es{\batch}| > \gamma|\es{v}|$ }
        {\label{adding-to-bucket}
            Mark $v$ as \good{}, add $v$ and its ancestors
            to $\batch$.\\
            Set $s=0$.\\
        }{
            Mark $v$ as \bad{}. $s = s+1$.\\ 
        }
        Remove $v$ from $K_i$.\\
    }
}
$\bucket_1, \ldots, \bucket_k \leftarrow$ Prune($H$, $B$, $\bucket_1$, \ldots, $\bucket_k$) [\cref{alg:pruning}].\label{line:pruning}
}
Return $\batch$.\\
\end{algorithm}

\begin{algorithm}[hbtp]
\caption{Prune$(H, B, \bucket_1, \ldots, \bucket_k)$}
\label{alg:pruning}
\LinesNumbered
\KwResult{New buckets $\bucket_1, \ldots, \bucket_k$.}
\KwIn{ A graph $H = (V, E)$, a batch $\batch \subseteq V$, and buckets $\bucket_1, \ldots, \bucket_k$.}
Obtain estimates $\estvh{v}{H}$ and $\esteh{v}{H}$ for all nodes $v \in \cup_{i=1}^k \bucket_i$ in the graph $H$.\\
Let $H' \leftarrow H \setminus \as{B}$\\
Obtain estimates $\estvh{v}{H'}$ and $\esteh{v}{H'}$ for all nodes $v \in \cup_{i=1}^k \bucket_i$ in the graph $H'$.\\
\For{$i = k$ to 1} {
    \For{each node $v$ in bucket $\bucket_i$}{
    $X \leftarrow \dfrac{\estvh{v}{H'}}
    {\estvh{v}{H}}$.\\
    $Y \leftarrow \dfrac{\esteh{v}{H'}}
    {\esteh{v}{H}}$.\\
    \If{$X \leq 2\gamma$ or 
    $Y \leq 2\gamma$}{\label{prune-condition}
        Remove $v$ from $\bucket_i$.\\
    }
}
}
Return the new buckets $\bucket_1, \ldots, \bucket_k$.
\end{algorithm}

\subsection{Analysis}

We first provide two key properties of the batch $\batch$ of vertices found by \cref{alg:core_algo} that are crucial for our final approximation factor and then analyze the running time of the algorithm. 
\ifsubmission
Due to space constraints, some proofs are relegated to~\cref{app:proofs}.
\fi

\ifsubmission
\myparagraph{Quality of the Schedule}\label{sec:quality}
\else
\subsubsection{Quality of the Schedule}\label{sec:quality}
\fi
We show that $\batch$ comprises of vertices whose ancestor sets do not overlap significantly, and further that it is the ``maximal'' such set.

\begin{lemma}\label{lem:good-schedule}
The batch $\batch$ returned by \cref{alg:core_algo} satisfies
$|\sv{\batch}| > \gamma \sum_{v \in \batch} 
|\sv{v}|$ and $|\sedge{\batch}| > \gamma \sum_{v \in \batch} |\sedge{v}|$.
\end{lemma}

\begin{proof}
Let $\batch^{(\ell)} \subseteq B$ denote the set containing the first $\ell$ vertices added to $B$ by the algorithm. We prove the lemma via induction. In the base case, $\batch^{(1)}$ consists of a single vertex and trivially satisfies the claim. Now suppose that the claim is true for some $\ell \geq 1$ and let $v$ be the $(\ell+1)$-th vertex to be added to $B$. 
By~\cref{adding-to-bucket} of~\cref{alg:core_algo},
we add a vertex $v$ into $\batch^{(\ell)}$ if and only
if $|\sv{v} \setminus \sv{\batch^{(\ell)}}| > \gamma 
|\sv{v}|$ and $|\sedge{v} \setminus \sedge{\batch^{(\ell)}}|
> \gamma |\sedge{v}|$. Furthermore,
since we enumerate $\sv{v}$
via DFS, our calculation of the cardinality of each of
these sets is exact. We now have, $|\sv{\batch^{(\ell+1)}}| = |\sv{\batch^{(\ell)}}| + |\sv{v} \setminus \sv{\batch}| > |\sv{\batch^{(\ell)}}| + \gamma 
|\sv{v}|$. By the induction hypothesis, we now have $|\sv{\batch^{(\ell+1)}}| > \gamma \sum_{w \in \batch^{(\ell)}} \sv{w} + \gamma \sv{v} = \gamma \sum_{w \in \batch^{(\ell+1)}} \sv{w}$. The same proof also holds for $\sedge{\batch}$ and the lemma follows.
\end{proof}

\begin{lemma}\label{lem:not-added-to-batch}
If a vertex $w$ was not added to $B$, it is pruned
by~\cref{alg:pruning}, with high probability.
If a vertex $v$ is pruned by~\cref{alg:pruning}, then
$|\sv{v} \setminus \sv{\batch}| \leq 4\gamma|\AncestorSet{v}|$
or $|\sedge{v} \setminus \sedge{\batch}| \leq 
4\gamma |\sedge{v}|$, 
with high probability.
\end{lemma}

\begin{proof}
We first prove 
that any vertex $v$ that is not added to $B$ must be 
removed from its bucket by \cref{alg:pruning}. Any vertex not added
to $B$ must have $|\AncestorSet{v} \setminus \AncestorSet{B}| \leq 
\gamma |\AncestorSet{v}|$. 
By~\cref{lem:ancestor-est-correct}, 
$\estvh{v}{H'} \leq 4/3 |\AncestorSet{v} \setminus \AncestorSet{B}|$ and $\estvh{v}{H} \geq 2/3|\AncestorSet{v}|$, with
high probability.
This must mean that $\frac{\estvh{v}{H'}}{\estvh{v}{H}} \leq \frac{4/3 |\AncestorSet{v} \setminus \AncestorSet{B}|}{2/3|\AncestorSet{v}|} \leq
\frac{4/3\gamma |\AncestorSet{v}|}{2/3|\AncestorSet{v}|} \leq 2\gamma$.
Thus, $v$ will be pruned. The same proof holds for $\esteh{v}{H'}$.

We now prove that the pruning procedure successfully prunes vertices with
not too many unique ancestors.
 In~\cref{alg:pruning}, by \cref{lem:ancestor-est-correct} (setting $\epsilon=1/3$), we have with high probability, $\estvh{v}{H'} \geq 2/3|\svh{v}{H'}|$.
Similarly, with high probability, $\estvh{v}{H} \leq 
4/3|\svh{v}{H}|$. 
This means
$X = \frac{\estvh{v}{H'}}{\estvh{v}{H}} \geq
\frac{2/3|\svh{v}{H'}|}{4/3|\svh{v}{H}|} = \frac{1}{2}
\left(\frac{|\svh{v}{H'}|}{|\svh{v}{H}|}\right)$.  
By the same argument, we also have $Y = \frac{\esteh{v}{H'}}{\esteh{v}{H}} \geq \frac{1}{2}
\left(\frac{|\sedgeh{v}{H'}|}{|\sedgeh{v}{H}|}\right)$ with high probability.

By~\cref{prune-condition} 
of~\cref{alg:pruning}, when we remove a vertex $v$ we have either $X \leq 2 \gamma$ or $Y \leq 2 \gamma$. 
By the above, $X, Y \geq \frac{1}{2}(4\gamma) = 2\gamma$. Thus, the
largest that $\left(\frac{|\svh{v}{H'}|}{|\svh{v}{H}|}\right)$ or
$\left(\frac{|\EdgeSetH{v}{H'}|}{|\EdgeSetH{v}{H}|}\right)$ can be
while still being pruned is $4\gamma$.
Thus, with high probability, we have either $\frac{|\svh{v}{H'}|}{|\svh{v}{H}|} \leq 4 \gamma$ or $\frac{|\sedgeh{v}{H'}|}{|\sedgeh{v}{H}|} \leq 4 \gamma$. Since $\svh{v}{H'} = \sv{v} \setminus \sv{\batch}$, the claim follows.
\end{proof}

The above two lemmas tell us that there are enough
unique elements in each batch $B$, any vertex
not added to $B$ will be pruned w.h.p., and the pruning
procedure only prunes vertices with a large enough overlap with $B$ w.h.p.
This allows us to show the following lemma on the length of the schedule
produced by~\cref{alg:small-subgraph} for small subgraph $H$.
We
first show that we only call~\cref{alg:core_algo} at most $O\left(\log_{1/\gamma}(\rho)\right)$
times from~\cref{line:num-batches} of~\cref{alg:small-subgraph}.

\begin{lemma}\label{lem:round-till-empty}
The number of batches needed to be scheduled
before all vertices in $H$ are scheduled is
at most $4 \log_{1/4\gamma}(2\rho)$, 
with high probability.
\end{lemma}

\begin{proof}
By~\cref{lem:not-added-to-batch}, each vertex $v$ we do not schedule in a batch $\batch$
has at least $(1 - 4\gamma)|\sv{v}|$ vertices in $\sv{\batch}$ or at least
$(1-4\gamma)|\sedge{v}|$ edges in $\sedge{\batch}$. 
Since we assumed that all vertices in $H$ have $\leq 2\rho$ ancestors, this means that
$v$ can only remain unscheduled for at most $2 \log_{1/ 4\gamma}(4\rho^2)$ batches
until $\sv{v}$ and $\sedge{v}$ both become empty ($\sgraph{v}$ can have at most $4\rho^2$ 
edges).
\end{proof}

Using~\cref{lem:round-till-empty}, we can prove the length of the schedule 
for $H$ using~\cref{alg:small-subgraph}. The proof of this lemma is similar to the
proof of schedule length of small subgraphs in~\cite{LR02}.

\begin{lemma}\label{lem:small-subgraph-processing}
With high probability, the schedule obtained from~\cref{alg:small-subgraph} has size at most 
$\frac{|V_H|}{\gamma M} + 12\rho\log_{1/4\gamma}(2\rho)$ on $M$ processors. 
\end{lemma}

\begin{proof}
By definition of the input, each $\sv{v}$ for $v \in V_H$
has at most $2\rho$ elements. 
Recall that we schedule all elements in each batch $\batch$ 
by duplicating the common shared ancestors such that
we obtain a set of independent ancestor sets to schedule. 
Then, we use a standard list scheduling
algorithm to schedule these lists; see~\cref{app:list-scheduling} for a classic list scheduling
algorithm.
Each vertex in $H$ gets scheduled in exactly one batch since
we remove all scheduled vertices from the subgraph used to compute the 
next batch. Let $B_1, B_2, \ldots, B_k$ denote the batches scheduled by \cref{alg:small-subgraph}.
Let $H_i$ be the subgraph obtained from $H$ by removing batches $B_0, \dots, B_{i - 1}$ and adjacent edges. ($B_0$ is empty.)
By~\cref{lem:good-schedule}, with high probability, for each batch $\batch_i$, we have $\sum_{v \in \batch_i} |\AncestorSetH{v}{H_i}| \leq \frac{1}{\gamma} |\AncestorSetH{\batch_i}{H_i}|$. Let $Z_i = \frac{1}{\gamma} |\AncestorSetH{\batch_i}{H_i}|$.

Graham's list scheduling algorithm \cite{graham:schedule} for independent jobs is known to produce a schedule whose length is at most the total length of jobs divided by the number of machines, plus the length of the longest job. In our case, we treat each ancestor set as one big independent job, and thus for each batch $\batch_i$, this bound becomes $Z_i/M + 2\rho$.

Finally \cref{alg:small-subgraph} inserts an idle time of $\rho$ between two successive batches. The total length of the schedule is thus upper bounded by (where $k$ is the number of batches): 
\begin{align*}
    \sum_{i=1}^k \left(\frac{Z_i}{M} + 2 \rho + \rho\right) &\leq 3 \rho \cdot 4 \log_{1/4\gamma}(2 \rho) \sum_{i=1}^k \frac{Z_i}{M} \;\;\; \text{(by~\cref{lem:round-till-empty})}\\
    &\leq 3 \rho \cdot 4 \log_{1/4\gamma}(2 \rho) + \frac{1}{\gamma M} \cdot \sum_{i=1}^k |\AncestorSetH{B_i}{H_i}|\\
    &= \frac{|V_H|}{\gamma M} + 12 \rho \log_{1/4\gamma}(2 \rho) \qedhere
\end{align*}
\end{proof}

\ifsubmission
\myparagraph{Running Time}\label{sec:running-time}
\else
\subsubsection{Runtime of Scheduling Small Subgraphs}
\fi
In order to analyze the running time of \cref{alg:core_algo}, we need a 
couple of technical lemmas. The key observation is that although computing 
the ancestor sets $\as{v}$ and $\es{v}$ (in \cref{line:ancestor}) of a 
vertex $v$ takes $O(|\es{v}|)$ time in the worst case, we can bound the 
total amount of time spent computing these ancestor sets by the size of the 
ancestor sets scheduled in the batch. There are two main components to the 
analysis. First, we show that after every iteration of the \emph{pruning} 
step, the number of vertices in each bucket reduces by at least a constant 
fraction and hence the sampling procedure is repeated at most $O(\ln n)$ 
times per batch. 
Secondly, we use a charging argument to upper bound the amount of time spent 
enumerating the ancestor sets of sampled vertices.

\smallparagraph{Finding Stale Vertices}
We first argue that with high probability, there are at most $O(\ln n)$ iterations of the while loop in \cref{line:while} of~\cref{alg:core_algo}. Intuitively, in each iteration of the while loop, the number of vertices in any bucket $\bucket_i$ reduces by at least a constant fraction. 

\ifsubmission
\later{
\subsection{Runtime of Scheduling Small Subgraphs}
\fi
\begin{lemma}\label{lem:large-number-bad}
For any constant $d \geq 1$ and $\psi > 0$, there is a constant $c \geq 1$ such that, with probability at least $1 - \frac{1}{n^d}$, at most a $\psi$-fraction of remaining nodes in each bucket are fresh after sampling $c \ln n$ stale vertices consecutively.
\end{lemma}

\ifsubmission
}

\later{
\fi
\begin{proof}
The main approach behind the proof is that we show that for 
any bucket where a constant fraction $\psi$ of the vertices in
the bucket are \good, for any $c \ln n$ consecutively sampled
vertices, we expect to see $\psi c \ln n$ \good vertices. 
Furthermore, we show a concentration bound around this expected
number of vertices using the Chernoff bound. 
Thus, we can conclude that if we
see $c \ln n$ \bad vertices consecutively (with no good vertices), 
then with high probability,
at most a small constant fraction of the remaining nodes in the bucket
is \good.

\cref{alg:core_algo} samples 
the vertices in each bucket $\bucket_i$ 
consecutively, uniformly at random 
without replacement, until a new \good vertex
is found or at least $c \ln n$ \bad vertices
are sampled consecutively. Let $F$ be the set of \good and
\bad vertices
sampled (and removed) so far from bucket $K_i$ \emph{up to the
most recent time} a \good vertex was sampled from $K_i$ (i.e.\ $F$
includes all vertices sampled including and up to the most recent
\good vertex sampled from $K_i$).
Let $\psi$ be some fraction $0 < \psi < 1$.
Suppose at most a $\psi$-fraction of the vertices in bucket 
$\bucket_i$ are \good after removing the previously sampled
$F$ vertices. 
Such an $\psi$ exists for every bucket with
at least one \good vertex and one \bad vertex after doing such 
removals. (In the case when
all vertices in the bucket are \good, all vertices from that bucket
will be sampled and added to $B$. If all vertices in the bucket are
\bad, then $c\ln n$ \bad vertices will be sampled immediately.)

From here on out, we assume the bucket $K_i$ only contains the
remaining vertices after the previously sampled $F$ vertices were 
removed. We assume the number of remaining
vertices in $K_i$ is more than 
$c \ln n$. The probability that each of the next sampled vertices
is a \good vertex is at least $\psi$.
The expected number of \good vertices 
in the $c \ln n$ samples
from $\bucket_i$ is lower bounded by:

\begin{align*}
    \sum_{i = 1}^{c \ln n} \left(i \cdot {c \ln n \choose i}\psi^{i} (1 - \psi)^{c\ln n - i}\right) = 
    \psi c \ln n.
\end{align*}

Suppose for our analysis that we only remove \bad vertices when we
sample them (and not \good vertices).
The above is a lower bound, in this setting, 
on the expected number of sampled
\good vertices from
$K_i$ since $\psi$ is the fraction of \good vertices after removing $F$;
if we remove more \bad vertices, the fraction of \good vertices 
cannot decrease so 
$\psi$ upper bounds the fraction of \good vertices
in $K_i$ as we remove more stale vertices. This assumption 
is the same as our algorithm when all $c \ln n$
sampled vertices are \bad.

By the Chernoff bound, the probability that we 
sample less than $(1 - \eps)\psi c\ln n$
\good vertices 
is less than $\exp\left(-\frac{\eps^2\psi c\ln 
n}{2}\right)$.
When $c > \frac{1}{(1-\eps)\psi}$, 
$(1-\eps)\psi c\ln n \geq 1$ 
for any $0 < \eps < 1$ and $0 < \psi < 1$. 
Then, the probability that no
\good vertices are sampled is less than
$\exp\left(-\frac{\eps^2\psi c\ln
n}{2}\right) = n^{\frac{-\eps^2 \psi c}{2}}$. 
It is easy to consider the case for constant $\psi \in (0, 1)$. If 
$\psi = o(1)$, then there exists a constant $\phi$ for which at most
a $\phi$-fraction of the vertices in $K_i$ are \good. If $\psi = \omega(1)$,
then the probability becomes super-polynomially small.
We can sample $c \ln n$ vertices for large enough 
constant $c \geq \frac{6d}{\eps^2 \psi}$ 
such that
with probability at least $1 - \frac{1}{n^d}$ 
for any constant $d \geq 1$, there
exists less than $\psi$-fraction of vertices
in the bucket that are \good if the next $c\ln n$ sampled
vertices are \bad. The factor of $6$ in the bound
$c \geq \frac{6d}{\eps^2 \psi}$ is useful when we take the 
union bound over multiple trials (at most $O(n^2)$)
for all buckets used during
the course of this algorithm.
\end{proof}
\ifsubmission
}
\fi

\begin{lemma}\label{lem:while-loop-times}
We perform $O(\ln n)$ iterations
of sampling and pruning, with high probability, 
before all buckets are empty. In other words,
with high probability,
\cref{line:while} of~\cref{alg:core_algo}
runs for $O(\ln n)$ iterations.
\end{lemma}

\begin{proof}
We prove the lemma for one bucket $\bucket_i$ and by
the union bound, the lemma holds for all buckets.
First, any sampled vertex which is \good is added to
$\batch$. Furthermore, we showed in~\cref{lem:not-added-to-batch} 
that any vertex which is
\bad is removed from $\bucket_i$ by~\cref{alg:pruning}.
Since the estimates $\ev{v}$ and $\ee{v}$ are within
a $\frac{1}{3}$-factor of $|\sv{v}|$ and $|\sedge{v}|$, 
respectively, we can upper bound $\frac{\estvh{v}{H'}}{
\estvh{v}{H}} \leq 2 \cdot \frac{|\sv{v} \setminus \sv{B}|}
{|\sv{v}|}$ (same holds for $\ee{v}$). If a vertex $v$ is
\bad, then with high probability, we have either $\frac{\estvh{v}{H'}}{
\estvh{v}{H}} \leq \frac{2|\sv{v} \setminus \sv{B}|}
{|\sv{v}|} \leq 2\gamma$ or $\frac{\esteh{v}{H'}}{
\esteh{v}{H}} \leq 2\gamma$, and it
is removed by~\cref{prune-condition} of~\cref{alg:pruning}.
Since any \good vertices that are sampled gets added
into $\batch$ and all \bad vertices are pruned at the end
of each iteration, it only remains to show a large enough
number of \bad vertices are pruned.

\cref{lem:large-number-bad} guarantees that, with high probability,  at least
$(1 - \psi)$-fraction of the vertices in $\bucket_i$
are \bad for any constant $\psi \in (0, 1)$. Then,~\cref{alg:pruning} removes
at least $(1 - \psi)|\bucket_i|$ vertices in
$\bucket_i$ in each iteration. The number of iterations
needed is then $\log_{1/(1 - \psi)}(|\bucket_i|) = O(\ln n)$.

Since there exists $O(\log \rho)$
buckets and $O(m)$ estimates, we
can take the union bound 
on the probability of 
success 
over all buckets and 
estimates. We obtain,
with high
probability, $O(\log n)$ iterations are necessary before
all buckets are empty. 
\end{proof}

\smallparagraph{Charging the Cost of Examining
\Bad Sets} 
Here we describe our charging argument that allows us to explictly enumerate 
the ancestor set of each sampled vertex. Computing the ancestor set of a 
vertex $v$ takes time $O(|\EdgeSet{v}|)$ using DFS. Since a fresh vertex 
gets added to the batch, the cost of computing the ancestor set of a fresh 
vertex can be easily bounded by the set of edges in $\EdgeSet{B}$, achieving a total cost, specifically, of $O\left(\frac{1}{\gamma}|\EdgeSet{B}|\right)$. Our 
charging argument allows us to bound the cost of computing ancestor sets of 
sampled stale vertices by charging it to the most recently sampled fresh 
vertex. Using the above, we provide the runtime of~\cref{alg:core_algo} below and then the runtime of~\cref{alg:small-subgraph}.

\iffull
We describe our charging 
argument
that allows us to look at $O(\log n)$ 
 additional vertices within each bucket until
 we find $O(\log n)$ consecutive vertices
 that are \bad. There are two parts to calculating the 
 runtime of enumerating the ancestor set of each 
 vertex that we sample via each iteration of the
 loop given in~\cref{line:forloop} of~\cref{alg:core_algo}.
 First, we calculate the runtime of enumerating the
 sets of all \good vertices. Then, we perform a
 charging argument that charges the runtime of enumerating
 the ancestor sets of \bad vertices to the cost of 
 enumerating the ancestor set of the \emph{most recently
 added vertex to $\batch$}. Together this allows
 us to show the following.

\else
\later{
\fi
\begin{lemma}\label{lem:charging-runtime}
With high probability,
the total runtime of enumerating the ancestor sets 
of all sampled vertices in
\cref{alg:core_algo} is $O\left(\frac{1}{\gamma}
|\sedge{\batch}|\ln \rho\ln n\right)$. 
In other words,
the total runtime of performing all iterations of
\cref{line:forloop} of~\cref{alg:core_algo} is
$O\left(\frac{1}{\gamma}
|\sedge{\batch}|\ln \rho\ln n\right)$, 
with high probability.
\end{lemma}

\begin{proof}
First, we calculate the runtime of enumerating the 
ancestor sets of each element of $\batch$.
By~\cref{lem:good-schedule}, 
$|\sedge{B}| \geq \gamma \sum_{v \in \batch}
|\sedge{v}|$. Hence, the amount of time to enumerate
all ancestor sets of every vertex in $B$ is
at most $\frac{1}{\gamma}|\sedge{B}|$.

We employ the following charging scheme to 
calculate the total time necessary to enumerate
the ancestor sets of all sampled \bad vertices.
Let $u$ be the most recent vertex added to 
$\batch$ from some bucket $\bucket_i$. We charge the cost
of enumerating the ancestor sets of all
\bad vertices sampled after $u$ to the cost
of enumerating the ancestor set of $u$.
Since we sample at most $O(\log n)$ consecutive
\bad vertices from each bucket before moving to the next bucket,
$u$ gets charged with at most the work of enumerating 
$O(\log \rho \log n)$ vertices from the same
or smaller buckets.
With high probability,
the largest ancestor set in 
bucket $\bucket_i$ has a size at 
most four times the smallest ancestor
set size. Since we sample vertices
in decreasing bucket size, we charge
at most $O(|\sedge{v}|\log \rho\log 
n)$ work to $v$. 

By our bound on the cost of enumerating all ancestor sets of vertices in $\batch$,
the additional charged cost results in a total cost of $\sum_{v \in B}
|\sedge{v}| \cdot O(\log \rho \log n) = O(\frac{1}{\gamma}|\sedge{B}|\log \rho\log n)$.
\end{proof}
\ifsubmission
}
\fi

\begin{lemma}\label{lem:runtime-core}
\cref{alg:core_algo} runs in $O\left(\frac{1}{\gamma} |\sedgeh{\batch}{H}| \ln \rho \ln n + |E_H| 
\ln^3 n\right)$ time, with high probability.
\end{lemma}

\begin{proof}
The runtime of~\cref{alg:core_algo} consists of three parts: the time to sample
and enumerate ancestor sets, the time to prune \bad vertices, and the time
to list schedule all vertices in $\batch$. 

By~\cref{lem:charging-runtime}, the time it takes to enumerate all sampled ancestor sets
is $O\left(\frac{1}{\gamma}|\sedge{\batch}|\log \rho\log n\right)$ over all iterations 
of the loops on~\cref{line:while} and~\cref{line:forloop} of~\cref{alg:core_algo}.

The time it takes to run~\cref{alg:pruning} is 
$O(|E_H|\estimationtime)$ since obtaining 
the estimates for each node (by~\cref{lem:ancestor-est-correct}), creating graph $H'$, and calculating $X$ and $Y$ for each
node in the bucket can be done in that time. By~\cref{lem:while-loop-times}, we perform
$O(\ln n)$ iterations of pruning, with high probability. 
Thus, the total time to prune the graph is $O(|E_H|\ln^3 n)$.
\end{proof}

\begin{lemma}\label{lem:subgraph-time}
Given a graph $H = (V_H, E_H)$ where $|\sv{v}| \leq 2\rho$ for each $v \in V_H$ and
parameter $\gamma \in (0, 1/2)$, the time it takes to compute the schedule of $H$ using
\cref{alg:small-subgraph} is, with high probability, $O\left(\frac{1}{\gamma}|E_H| \ln \rho \ln n + |E_H| \ln_{1/4\gamma}\rho\ln ^3 n + |V_H|\ln M\right)$.
\end{lemma}

\begin{proof}
By~\cref{lem:round-till-empty}, we perform $O(\log_{1/4\gamma}\rho)$ calls to~\cref{alg:core_algo}.
Each call to~\cref{alg:core_algo} requires $O\left(\frac{1}{\gamma} |\sedgeh{\batch}{H}| \ln \rho \ln n + |E_H| \ln^3 n\right)$ time
by~\cref{lem:runtime-core}. However, we know that each vertex (and edges adjacent to it) is
scheduled in exactly one batch.

For each batch $\batch$,
our greedy list scheduling procedure schedules each $\sv{v}$
for $v \in \batch$
greedily and independently by duplicating vertices that appear in more than one ancestor set.
Thus, enumerating all the ancestor sets require $O\left(\frac{1}{\gamma}|\sedge{\batch}|\right)$ time
by~\cref{lem:good-schedule}.
When $M > |\batch|$, we easily
schedule each list on a separate machine in $O\left(\frac{1}{\gamma}|\sedge{\batch}|\right)$
time.
Otherwise, to schedule the lists,
we maintain a priority queue of the machine finishing
times. For each list, we greedily assign it to the machine
that has the smallest finishing time. We can perform
this procedure using $O(M \ln M)$ time. Since $M \leq |\batch|$, this results in $O(|\batch| \ln |\batch|)$ time
to assign ancestor sets to machines.

Thus, the total runtime of all calls to~\cref{alg:core_algo}
is

\begin{align*}
    &\sum_{i = 1}^{\log_{1/4\gamma}\rho} O\left(\frac{1}{\gamma} |\sedgeh{\batch_i}{H}| \ln \rho \ln n + |E_H| \ln^3 n + \frac{1}{\gamma}|\EdgeSet{B_i}| + |B| \ln |B|\right) \\
    &= O\left(\frac{1}{\gamma} |E_H| \ln \rho \ln n + |E_H| \ln_{1/4\gamma}\rho\ln ^3 n\right).
\end{align*}

Then, the time it takes to perform~\cref{line:comm-delay}
of~\cref{alg:small-subgraph} is $O(1)$ per iteration. Scheduling
vertices with no adjacent edges requires $|B| \ln |B| = O(|V_H| \ln M)$
time.
Finally, the time it takes to remove each $v \in \sv{\batch}$ and all edges adjacent to $v$ from
$H$ for each batch $\batch$ is $O(|V_H| + |E_H|)$. Doing this for $O(\ln_{1/4\gamma}\rho)$ iterations
results in $O(|V_H| + |E_H|\ln_{1/4\gamma}\rho)$ time.
\end{proof}

\section{Scheduling General Graphs}\label{sec:general}

We now present our main scheduling algorithm for scheduling any DAG $G = (V, E)$
\ifsubmission
(the full pseudocode is included in Appendix C  in our full paper~\cite{fullversion}).
\else
(\cref{alg:main-alg} in~\cref{sec:alg-full}).
\fi
This
algorithm also uses as
a subroutine the procedure for estimating
the number of ancestors of each vertex in $G$ as
described in~\cref{sec:estimators}.
We use the estimates to compute the small subgraphs which
we pass into~\cref{alg:small-subgraph} to schedule.
We produce the small subgraphs by setting the cutoff for
the estimates to be
$\frac{4}{3}\rho$. This produces small graphs where the number
of ancestors of each vertex is upper bounded by $2\rho$, with
high probability. We present a simplified algorithm
below in~\cref{alg:main-alg-short}.
\ifsubmission
The full pseudocode for our main algorithm is given in Algorithm 9
 in our full paper~\cite{fullversion}.
\fi

\begin{algorithm}[htb]
\caption{ScheduleGeneralGraph$(G)$}
\label{alg:main-alg-short}
\LinesNumbered
\SetAlgoLined
\KwResult{A schedule of the input graph $G = (V, E)$ on $M$ processors.}
\KwIn{A directed acyclic task graph $G = (V, E)$.}
Let $\mathcal{H} \leftarrow \emptyset$ represent a list of small subgraphs that we will build.\\
\While{$G$ is not empty} {
    Let $V_H$ be the set of vertices in $G$ where $\estv{v} \leq \frac43\rho$ for each $v \in V_H$.\\
    Compute edge set $E_H$ to be all edges induced by $V_H$.\\
    Add $H = (V_H, E_H)$ to $\mathcal{H}$.\\
    Remove $V_H$ and all incident edges from $G$.\\
}
\For{$H \in \mathcal{H}$ in the order they were added}{\label{line:subgraph}
Call ScheduleSmallSubgraph($H$) to obtain
a schedule of $H$.
[\cref{alg:small-subgraph}]\\
}
\end{algorithm}

\myparagraph{Quality of the Schedule and Running
Time}
Let $\opt$ be the length of the optimal schedule. We first give two bounds on $\opt$, and then relate them to the length of the schedule found by our algorithm. A detailed set of
proofs is provided in~\cref{app:main-alg}.

Our main algorithm, \cref{alg:main-alg-short}, partitions the vertices of $G$ into small subgraphs $H \in \mathcal{H}$. It does so based on estimates of ancestor set sizes. We first lower bound $\opt$ by
working with exact ancestor set sizes. Since the schedule
produced by our algorithm cannot have length smaller than $\opt$,
this process also provides a lower bound on our schedule length. Then, we show that~\cref{alg:main-alg-short} does
not output more small subgraphs than the number of subgraphs produced
by working with exact ancestor set sizes, with high probability.

The crucial fact in obtaining our final runtime is that
producing the estimates of the number of ancestors of each
vertex requires $\tO(|V| + |E|)$ time \emph{in total} over
the course of finding all small subgraphs. Together, these
facts allow us to obtain~\cref{thm:final-runtime}.

\ifsubmission
\later{
\fi
\subsection{Quality of the Schedule Produced by the Main Algorithm}\label{app:main-alg}
Assuming we are working with exact ancestor set sizes, we
would wind up with vertex sets $V_1 \triangleq \{v\in V: |\AncestorSet{v}| \leq \rho\}$ and, inductively, for $i > 1$, $V_i \triangleq \{v\in V \setminus \bigcup_{j=1}^{i-1} V_j : |\AncestorSet{v} \setminus \bigcup_{j=1}^{i-1} V_j| \leq \rho\}$. Let $L$ be the maximum index such that $V_L$ is nonempty.

The following
lemma follows a similar argument as that found in Lepere-Rapine~\cite{LR02} (although we have simplified the analysis).
We repeat it here for completeness.

\begin{lemma}\label{lem:lb1}
  $\opt \geq (L-1)\rho$.
\end{lemma}

\begin{proof}
  We show by induction on $i$ that in any valid schedule, there exists a job $v \in V_i$ that cannot start earlier than time $(i-1) \rho$.  Given that, the job in $V_L$ starts at time at least $(L-1)\rho$ in $\opt$, proving the lemma.

  The base case of $i=1$ is trivial. For the induction step, consider a job $v \in V_{i+1}$. This job has at least $\rho$ ancestors in $V_i$ (call this set $A = \AncestorSet{v} \cap V_i$),
  since if it had less, $v$ would be in $V_i$ itself. All jobs in $A$ start no earlier than $(i-1)\rho$ by the induction hypothesis. There are two cases. If all of the jobs in $A$ are executed on the same machine as $v$, then it would take at least $\rho$ units of time for them to finish before $v$ can start. If at least one job in $A$ is executed on a different machine than $v$, then it would take $\rho$ units of time to communicate the result. In either case, $v$ would start later than the first job in $A$ by at least $\rho$, and thus no earlier than $i \cdot \rho$.
\end{proof}

\begin{lemma}\label{lem:lb2}
  $\opt \geq |V| / M$.
\end{lemma}

\begin{proof}
  Every job has to be scheduled on at least one machine, and the makespan is at least the average load on any machine.
\end{proof}

We show that our general algorithm only calls the schedule small subgraph procedure at most $L$ times, w.h.p.

\begin{lemma}\label{lem:numiter}
  With high probability,
  \ifsubmission
  \cref{alg:main-alg-short}
  \else
  \cref{alg:main-alg}
  \fi
  calls \cref{alg:small-subgraph} at most $L$ times on input graph $G = (V, E)$.
\end{lemma}

\begin{proof}
  By construction, the $V_i$'s are inductively defined by stripping
  all vertices with ancestor sets at most $\CommunicationDelay{}$ in
  size.
  With high probability, our estimates $\ev{v}$ are at most $\frac43
  |\sv{v}|$.
  \ifsubmission
  \cref{alg:main-alg-short}
  \else
  \cref{line:cutoff} of~\cref{alg:main-alg}
  \fi
  only takes
  vertex $v$ into the subgraph $H$ if $\estv{v} \leq \frac43\rho$.
  By~\cref{lem:ancestor-est-correct}, $\frac23|\AncestorSet{v}|
  \leq \estv{v} \leq \frac43 |\AncestorSet{v}|$. Then,
  $|\AncestorSet{v}| \leq \frac32\estv{v} \leq \frac32 \cdot \frac43 \rho = 2\rho$. Furthermore, since $|\AncestorSet{v}| \geq \frac34\cdot\estv{v}$, if $\estv{v} = \frac43\rho$, then $|\AncestorSet{v}| \geq \rho$. Hence, all vertices with height
  $\rho$ are added into $H$, with high probability.
  Taken together, this means that all vertices of $V_i$ (even if
  their ancestor sets
  are maximally overestimated) are contained in the small graphs $H$
  produced by iterations one through $i$ of \cref{line:subgraph} of
  \ifsubmission
  \cref{alg:main-alg-short}
  \else
  \cref{alg:main-alg}.
  \fi
  Since $V_L$ was chosen to be the last
  non-empty set, we know our algorithm runs for at most $L$
  iterations, with high probability.
\end{proof}

\begin{theorem}
\label{thm:main-schedule-length}
\ifsubmission
\cref{alg:main-alg-short}
\else
  \cref{alg:main-alg}
 \fi
  produces a schedule of length at most $O\left(\frac{\ln \rho}{\ln \ln \rho}\right) \cdot (\opt + \rho)$.
\end{theorem}

\begin{proof}
  In~\Cref{alg:small-subgraph}, by~\cref{lem:small-subgraph-processing}, the schedule length obtained from any small subgraph $H$ is  $\frac{|V_H|}{\gamma M} + 12 \CommunicationDelay{} \log_{1/4\gamma}(2\rho)$.

Let $\subgraphset$ be the set of all small subgraphs
\ifsubmission
\cref{alg:main-alg-short}
\else
\cref{alg:main-alg}
\fi
sends to~\cref{alg:small-subgraph} to be scheduled. By \cref{lem:numiter}, there are at most $L$ of them. Each vertex trivially appears in at most one subgraph.
Then the total length of our schedule is given by
\begin{align*}
    \sum_{H \in \subgraphset} \left(\frac{|V_H|}{\gamma M}
    + 12\rho\log_{1/4\gamma}(2\rho)\right)
    & = \sum_{H \in \subgraphset} \frac{|V_H|}{\gamma M} + \sum_{H \in \subgraphset} 12\rho\log_{1/4\gamma}(2\rho)\\
    &\leq \frac{|V|}{\gamma M} + L\left(12\rho\log_{1/4\gamma}(2\rho)\right).
\end{align*}

By Lemmas \ref{lem:lb1} and \ref{lem:lb2}, this last quantity is upper bounded by
$$
OPT \cdot \left(\frac{1}{\gamma} + 12\log_{1/4\gamma}(2\rho)\right) + \rho\cdot12\log_{1/4\gamma}(2\rho)
$$
Setting $\gamma = 1/\sqrt{\ln \rho}$ gives our bound of $(OPT+\rho) \cdot O(\frac{\ln \rho}{\ln \ln \rho})$.
\ignore{
By setting $\gamma = 1/\sqrt{\ln \rho}$,
we obtain that the maximum schedule length is
$\frac{|V|}{M} \cdot O(\sqrt{\ln \rho})
+ L\rho \cdot O(\frac{\ln \rho}{\ln \ln \rho}) \leq (OPT+\rho)
\cdot O(\frac{\ln \rho}{\ln \ln \rho})$, using Lemmas \ref{lem:lb1} and \ref{lem:lb2}.}
\end{proof}
\ifsubmission
}
\fi

\iffull
\subsection{Running Time of the Main Algorithm}
We prove the following theorem
regarding the runtime of~\cref{alg:main-alg}
which uses \cref{alg:small-subgraph} as a subroutine.
\fi

\begin{theorem}\label{thm:final-runtime}
On input graph $G = (V, E)$,
\ifsubmission
\cref{alg:main-alg-short}
\else
\cref{alg:main-alg}
\fi
produces
a schedule of length at most $O\left(\frac{\ln \rho}{\ln\ln \rho} \cdot (\opt + \rho)\right)$ and
runs
in time $O\left(n \ln M + \frac{m \ln^3 n \ln\rho}{\ln\ln\rho}\right)$, with
high probability.
\end{theorem}

\begin{proof}
By~\cref{lem:numiter},~\cref{alg:core_algo} is called at most $L$ times.
Then, since each vertex is in at most one small subgraph (and hence
each edge is in at most one small subgraph), the total runtime
for all the calls (by~\cref{lem:subgraph-time}) is
\begin{align*}
    &\sum_{i = 1}^L O\left(\frac{1}{\gamma}|E_{H_i}| \ln \rho \ln n + |E_{H_i}| \ln_{1/4\gamma}\rho\ln ^3 n + |V_{H_i}| \ln M\right) \\
    &= O\left(\frac{1}{\gamma}|E|\ln \rho \ln n + |E| \ln_{1/4\gamma}\rho\ln ^3 n + |V| \ln M\right).
\end{align*}

Furthermore, each iteration of~\cref{line:subgraph} requires estimating
$\ev{v}$ for a set of vertices $v$, adding $v$ to $H$, and checking all
successors of $v$. First, we show that $\ev{v}$ is computed at most
twice for each vertex in $V$, and then, we show that the rest of the steps
are efficient.

Each vertex contained in the queue, $Q$, in~\cref{alg:core_algo}, either does not have any ancestors,
or all of its ancestors are in $H$ (the current subgraph). If a
vertex $v \in Q$ was not added to $H$ during iteration $i$, then
it must have at least one ancestor in iteration $i$
and no ancestors in iteration $i + 1$. Since $v$ has no ancestors
in iteration $i + 1$, it must be added to
$H_{i + 1}$. The time it takes to compute the
estimate for one vertex is $O(\ln^2 n)$.
Thus, the total time it takes to compute
the estimate of the number of ancestors of all vertices is $O\left(m \estimationtime\right)$.
Adding $v$ to $H$ and checking all successors can be done in $O(m)$ time in total across all vertices and subgraphs.
Finally removing each $v \in H$ from $G$ can be done in $O(m)$ time in total for all $v$.

As earlier, we use $\gamma = \sqrt{\ln \rho}$,
so the total runtime summing the above
can be upper bounded by $O\left(n\ln M + \frac{m \ln^3 n \ln\rho}{\ln\ln\rho}\right)$.
Thus, the algorithm produces
a schedule of length $O\left(\frac{\ln \rho}{\ln\ln \rho} \cdot (\opt + \rho)\right)$ (by~\cref{thm:main-schedule-length})
and the total runtime of the algorithm is
$O\left(n\ln M + \frac{m \ln^3 n \ln\rho}{\ln\ln\rho}\right)$,
with high probability.
\end{proof}

\cref{thm:final-runtime} gives the main
result of our paper stated informally in~\cref{thm:1}
of the introduction.

\appendix

\iffull
\FloatBarrier
\fi
\section{Count-Distinct Estimator~\cite{BJKST02}}\label{sec:estimator}

\ifsubmission
\else
\begin{algorithm}[htbp]
\SetAlgoLined
\KwResult{An estimate on the number of distinct elements in the input multiset.}
\caption{\cite{BJKST02} algorithm for estimating number of distinct elements in a
set.}\label{alg:distinct-elems}
 \KwIn{A multiset $S$ of elements where $n = |S|$.}
Let $t \leftarrow \frac{c}{\eps^2}$ for some fixed constant $c \geq 1$.\\
Let $\mathcal{H}$ be a 2-universal hash family
where each $h_i \in \mathcal{H}$ is a hash function $h_i: [n] \rightarrow [N]$ where $N = n^3$.\\
Let $h_1, \dots, h_{c'\log n} \in \hash$ be a set of $c'\log n$ hash functions
chosen uniformly at random from $\mathcal{H}$ where $c' \geq 1$ is a constant.\\
Maintain a binary tree $T_i$ for each $h_i$
of the smallest $t$ values seen so far of the
outputs of $h_i$. Initially each $T_i$ has
no elements.\\
\For{$a_j \in S$}{
    \For{$h_i \in [h_1, \dots, h_{c'\log n}]$}{
    Compute $h_i(a_j)$.\\
    \If{$h_i(a_j)$ is smaller than the largest element in $T_i$}{
        Add $h_i(a_j)$ to $T_i$.\\
        \If{size of $T_i$ is greater than $t$}{
        Remove the largest element in $T_i$ from $T_i$.}
        }
    }
}
Initialize an empty list $L$.\\
\For{each $T_i$}{
Let $\ell_i$ be the largest element in $T_i$.\\
Insert $\ell_i$ into $L$.
}
Sort $L$.\\
Return $tN/\ell$ where $\ell$ is the median of $L$.
\end{algorithm}
\fi
We provide the algorithm of Bar-Yossef et al.\ \cite{BJKST02}
\ifsubmission
in Appendix A  in our full paper~\cite{fullversion}.
\else
in~\cref{alg:distinct-elems}.
\fi
The algorithm of~\cite{BJKST02} works as follows.
Provided a multiset $S$ of elements where $n = |S|$,
we pick $t = \frac{c}{\eps^2}$ where $c$ is some fixed constant
$c \geq 1$ and $\mathcal{H}$, a $2$-univesal hash family.
Then, we choose $O(\log n)$ hash functions from $\mathcal{H}$
uniformly at random,
without replacement. For each hash
function $h_i: [n] \rightarrow [n^3]$
($i = O(\log n)$), we maintain a balanced
binary tree $T_i$ of the \emph{smallest} $t$ values seen
so far from the hash outputs of $h_i$. Initially, all $T_i$ are
empty.
We iterate through $S$ and for each $a_j \in S$, we
compute $h_i(a_j)$ using each $h_i$ that we picked; we update
$T_i$ if $h_i(a_j)$ is smaller than the largest element in $T_i$
or if the size of $T_i$ is smaller than $t$.
After iterating through all of $S$, for each $T_i$, we add
the largest value of each tree $T_i$ to a list $L$. Then, we
sort $L$ and find the median value $\ell$ (using
the \emph{median trick}). We return $t n^3/\ell$
as our estimate.

We now show how to use
\ifsubmission
Bar-Yossef et al.~\cite{BJKST02}
\else
\cref{alg:distinct-elems}
\fi
to get our desired mergeable estimator. Let
$\mathcal{T}_X$
be the set of trees $T_i \in \mathcal{T}_X$ maintained
for the estimator defined
\ifsubmission
by Bar-Yossef et al.~\cite{BJKST02}
\else
by~\cref{alg:distinct-elems}
\fi
for multiset $X$.
Since each $T_i$ has size at most
$O(t) = O\left(\frac{1}{\eps^2}\right)$, the total space required to store all $T_i$ is $O\left(\frac{1}{\eps^2}\log^2 n\right)$ in bits.
We can initialize our estimator on input $d$ by picking
a set of random hash functions: $h_1, \dots, h_{d\log n} \in \mathcal{H}$. Let $H$ be the set of picked hash function. Then, for each set $S$, we initialize
$d\log n$ trees $T_i \in \mathcal{T}_S$
\ifsubmission
\else
(as used in~\cref{alg:distinct-elems})
\fi
and maintain $\mathcal{T}_S$ in memory.
The elements of $T_i$ are computed using $h_i \in H$.
Let $\dist_S$ denote the estimator for $S$.
Using $\mathcal{T}_S$ for set $S$, we can implement the following functions (pseudocode for the three functions can be found in
\cref{alg:estimator-additional-funcs}):

\begin{itemize}
    \item \textbf{$\est$.insert($\dist_{S}, x$)}: Insert $h_i(x)$ into $T_i \in \mathcal{T}_S$ for each $i \in [d\log n]$. If
    $T_i$ has size greater than $t$, delete the largest element of $T_i$.
    \item \textbf{$\est$.merge($\dist_{S_1}$, $\dist_{S_2}$)}:
    Here we assume that the same set of hash functions
    are used for both $\dist_{S_1}$ and $\dist_{S_2}$.
    For each pair of $T_{1, i} \in \mathcal{T}_{S_1}$ and $T_{2, i} \in \mathcal{T}_{S_2}$ for
    hash function $h_i$, build a new tree $T_i$ by taking the $t$ smallest elements from $T_{1, i} \cup T_{2, i}$.
    \item \textbf{$\est$.estimateCardinality($\dist_S$)}: Let $\ell$ be the median value of the largest values of the
    trees $T_i \in \mathcal{T}_S$. Return $tN/\ell$.
\end{itemize}

\begin{algorithm}[htb!]
\SetAlgoLined
\caption{Initialize New $\est$.}\label{alg:estimator-additional-funcs}
\KwIn{$\mathcal{D}_S, \mathcal{T}_S, t, h_i \in \mathcal{H}$ are as defined above
for multiset $S$. Let $n = |S|$.}
\textbf{$\est$.insert($\dist_{S}, x$):}\\
\Indp\For{$T_i \in \mathcal{T}_S$}{
Compute $h_i(x)$.\\
\If{$T_i$ has less than $t$ elements or
$h_i(x)$ is smaller than the largest
value in $T_i$}{Insert $h_i(x)$ into $T_i$.}
\If{$T_i$ has more than $t$ elements}{
Remove the largest element in $T_i$.\\}
}

\Indm\textbf{$\est$.merge($\dist_{S_1}$, $\dist_{S_2}$):}\\
\Indp\For{$T_{1, i} \in \mathcal{T}_{S_1}$, $T_{2, i} \in \mathcal{T}_{S_2}$}{
Perform inorder traversal of $T_{1, i}$ and $T_{2, i}$ to obtain
non-decreasing lists of elements, $L_{1, i}$ and $L_{2, i}$.\\
Merge $L_{1, i}$ and $L_{2,i}$ to obtain a new non-decreasing
list of elements, $L$.\\
Build a new balanced binary tree from the first $t$ elements of
$L$.\\
}

\Indm \textbf{$\est$.estimateCardinality($\dist_S$)}:\\
\Indp\For{$T_i \in \mathcal{T}_S$}{
    Insert largest element of $T_i$ into list $L$.
}
Sort $L$.\\
Let $\ell$ be the median of $L$.\\
Return $t n^3/\ell$.
\end{algorithm}

The estimator provided in Bar-Yossef et al.\ \cite{BJKST02} satisfies the following lemmas as proven in~\cite{CG06} (specifically it is proven that the estimator is unbiased):

\begin{lemma}[\cite{BJKST02}]\label{lem:desired-estimator}
The Bar-Yossef et al.~\cite{BJKST02} estimator is an $\left(\eps, \frac{1}{n^d}, O\left(\frac{1}{\eps^2}\log^2 n\right)\right)$-estimator
for the count-distinct problem.
\end{lemma}

\begin{lemma}\label{lem:merge-time}
Furthermore, the insert, merge, and estimate cardinality functions of the Bar-Yossef et al.~\cite{BJKST02} estimator can be implemented in $O\left(\frac{1}{\eps^2}\log^2 n\right)$ time.
\end{lemma}

\begin{proof}
\textbf{$\est$.insert} requires $O(\log t)$ time to insert $h_i(x)$
and $O(\log t)$ time to remove the largest element. Thus, this
method requires $O\left(\log\left(\frac{1}{\eps}\right)\right) =
O\left(\frac{1}{\eps^2}\right)$
time.
\textbf{$\est$.merge} requires $O(t) =
O\left(\frac{1}{\eps^2}\right)$
time to merge $T_{1, i}$ and
$T_{2, i}$ and also $O(t)$ time to build the new tree.
Finally, \textbf{$\est$.estimateCardinality} requires
$O(\log n)$ time to create the list $L$ and $O(\log n\log \log n)$
time to sort and find the median.
\end{proof}

\myparagraph{Estimating the Number of Ancestors and Edges}

Using the count-distinct estimator described above, we
can provide our full algorithms for estimating the
number of ancestors and the number of edges in the induced
subgraph of every vertex in a given input graph.

Our complete algorithm for estimating the number
of ancestors $\estv{v}$ of every vertex in an input graph
is given in~\cref{alg:ancestor-estimation}. Our algorithm for finding $\este{v}$ for every vertex
in the input graph is given
\ifsubmission
in Algorithm 8  in our full paper~\cite{fullversion}.
\else
in~\cref{alg:ancestor-edge-estimation}.
\fi

\begin{algorithm}[hbt!]
\caption{
Estimate Number of Ancestors}\label{alg:ancestor-estimation}
\SetAlgoLined
\KwResult{Estimate $\estvh{v}{H}$ such that  $(1-\epsilon)|\svh{v}{H}| \leq \estvh{v}{H} \leq (1+\epsilon)|\svh{v}{H}|$, $\forall v \in V$.}
\textbf{Input: } A graph $H = (V, E)$.\\
Topologically sort all the vertices in $H$.\\
\For{vertex $w$ in the topological order of vertices in $H$}{
Let $\pred(w)$ be the set of predecessors of $w$. \\
Let $\dist_w \leftarrow $ New $\left(\eps, \frac{1}{n^d},
O\left(\frac{1}{\eps^2}\log^2 n\right)\right)$-$\est$  for $w$.\\
$\est$.insert$(\dist_w, w)$\\
\For{$v \in \pred(w)$}{
        $\dist_w = \est$.merge$(\dist_w, \dist_v)$.\\
}
$\estvh{w}{H} = \est$.estimateCardinality$(\dist_w)$
}
\end{algorithm}

\ifsubmission
\else
\begin{algorithm}[hbt!]
\caption{
Estimate Number of Ancestor Edges}\label{alg:ancestor-edge-estimation}
\SetAlgoLined
\KwResult{Estimate $\esteh{v}{H}$ such that  $(1-\epsilon)|\sedgeh{v}{H}| \leq \esteh{v}{H} \leq (1+\epsilon)|\sedgeh{v}{H}|$, $\forall v \in V$.}
\textbf{Input: } A graph $H = (V, E)$.\\
Topologically sort all the vertices in $H$.\\
\For{vertex $w$ in the topological order of vertices in $H$}{
Let $\pred(w)$ be the set of predecessors of $w$. \\
Let $\dist_w \leftarrow $ New $\left(\eps, \frac{1}{n^d},
O\left(\frac{1}{\eps^2}\log^2 n\right)\right)$-$\est$  for $w$.\\
\For{$v \in \pred(w)$}{
        $\est$.insert$(\dist_w, (v,w))$\\
        $\dist_w = \est$.merge$(\dist_w, \dist_v)$.\\
}
$\esteh{w}{H} = \est$.estimateCardinality$(\dist_w)$
}
\end{algorithm}
\fi

\section{List Scheduling}\label{app:list-scheduling}
Here we provide a brief description of the classic Graham list scheduling algorithm~\cite{graham72}.
For our purposes, we are
given a set of vertices and their ancestors. We duplicate
the ancestors for each vertex $v$ so that each vertex and
its ancestors is scheduled as a \emph{single} unit
with job size equal to the \emph{number of ancestors} of
$v$. Then, we perform the following greedy procedure:
for each vertex $v$, we sequentially assign $v$ to the machine $M_i \in \mathcal{M}$ with \emph{smallest load} (i.e.\ load is
defined by the jobs lengths of all jobs assigned to it).
We can maintain loads of the machines in a heap to
determine the machine with the lowest load at any time.
To schedule $n$ jobs using this procedure requires
$O(n\ln M)$ time.
\FloatBarrier

\ifsubmission
\else
\section{Scheduling General Graphs Full Algorithm}\label{sec:alg-full}

\cref{alg:main-alg-short} is a shortened version of
\cref{alg:main-alg} presented here.

\begin{algorithm}[htb!]
\caption{ScheduleGeneralGraph$(G)$}\label{alg:main-alg}
\LinesNumbered
\SetAlgoLined
\KwResult{A schedule of the input graph $G = (V, E)$ on $M$ processors.}
\KwIn{A directed acyclic task graph $G = (V, E)$.}
Let $\mathcal{H} \leftarrow \emptyset$ represent a list of small subgraphs that we will build.\\
Initialize a data structure to keep track of marked vertices; all vertices are initially unmarked.\\
\While{$G$ is not empty} {
    Let $V_H \leftarrow \emptyset$ be the vertices of the current small subgraph we are building.\\
    Initialize queue $Q$ with all
    sources (vertices with no ancestors) of $G$.\\
    \While{$Q$ is not empty}{
        Remove first element $v$ of $Q$.\\
        Compute the size estimate for $v$'s ancestor set as $\ev{v}$.\\
        \If{$\ev{v} \leq \frac43 \rho$}{\label{line:cutoff}
            Mark $v$ and add it to $V_H$.\\
            \For{each successor $w$ of $v$}{
                \If{all predecessors of $w$ are marked}{
                    Enqueue $w$ into queue $Q$.\\
                }
            }
        }
    }
    Compute edge set $E_H$ to be all edges induced by $V_H$.\\
    Add $H = (V_H, E_H)$ to $\mathcal{H}$.\\
    Remove $V_H$ and all incident edges from $G$.\\
}
\For{$H \in \mathcal{H}$ in the order they were added}{\label{line:full-subgraph}
Call ScheduleSmallSubgraph($H$) to obtain
a schedule of $H$.
[\cref{alg:small-subgraph}]\\
}
\end{algorithm}
\fi
\FloatBarrier

\ifsubmission
\section{Deferred Proofs}\label{app:proofs}

In this section, we include all of the proofs deferred from the main
text.

\magicappendix
\fi

\ifsubmission
\bibliographystyle{plainurl}
\else
\bibliographystyle{alpha}
\fi
\bibliography{ref}

\end{document}